\documentclass[a4paper,12pt,leqno]{amsart} 
\usepackage{amsmath,amsthm,amssymb}  
\usepackage{mathdots}

\numberwithin{equation}{section}

\theoremstyle{plain}
\newtheorem{theorem}{Theorem}[section]
\newtheorem{proposition}[theorem]{Proposition}         
\newtheorem{corollary}[theorem]{Corollary} 
\newtheorem{lemma}[theorem]{Lemma} 
\newtheorem{definition}[theorem]{Definition}   

\theoremstyle{definition}  
 
\newtheorem{remark}[theorem]{Remark}  

\newcommand{\C}{\mathbb C}   
\newcommand{\R}{\mathbb R}
\newcommand{\Z}{\mathbb Z}

\renewcommand{\P}{\mathbb P}

\newcommand{\al}{\alpha}
\newcommand{\be}{\beta} 
\newcommand{\ga}{\gamma}
\newcommand{\de}{\delta}
\newcommand{\la}{\lambda}
\newcommand{\si}{\sigma} 
 
\newcommand{\La}{\Lambda}
\newcommand{\eps}{\epsilon}
\newcommand{\Om}{\Omega}
\newcommand{\De}{\Delta}
\renewcommand{\th}{\theta}
\newcommand{\om}{\omega}

\DeclareMathOperator{\diag}{diag}

\DeclareMathOperator{\Fix}{Fix}

\newcommand{\calA}{\mathcal{A}}
\newcommand{\calB}{\mathcal{B}}

\newcommand{\no}{\noindent}
\newcommand{\sub}{\subseteq}

\newcommand{\pr}{\prime} 
 
\newcommand{\st}{\ \vert\ }   
\renewcommand{\ll}{\lq\lq}
\newcommand{\rr}{\rq\rq\ }
\newcommand{\rrr}{\rq\rq} 
\newcommand{\lan}{\langle}
\newcommand{\ran}{\rangle} 
\renewcommand{\b}{\partial}
\newcommand{\zbar}{  {\bar z}  }
\newcommand{\zzb}{ {z\bar z}  }

\newcommand{\bp}{\begin{pmatrix}} 
\newcommand{\ep}{\end{pmatrix}}

\newcommand{\GL}{\textrm{GL}}
\newcommand{\SL}{\textrm{SL}}
\renewcommand{\O}{\textrm{O}}
\newcommand{\SO}{\textrm{SO}}

\newcommand{\SU}{\textrm{SU}}

\newcommand{\GLNR}{\GL_n\R}

\newcommand{\SLPC}{\SL_{n+1}\C}

\newcommand{\su}{{\frak s\frak u}}

\renewcommand{\sl}{\frak s\frak l}

\newcommand{\slpc}{\frak s\frak l_{n+1}\C}

\newcommand{\g}{{\frak g}}

\newcommand{\stoda}{$tt^\ast$\!-\!Toda lattice\ }
\newcommand{\sstoda}{$tt^\ast$\!-\!Toda lattice}
\renewcommand{\S}{S}

\newcommand{\smallsum}{  {\textstyle \text{$\sum$\,} }}
\renewcommand{\smallint}{  {\textstyle \text{$\int$} }}
\newcommand{\equ}{\!=\!}

\renewcommand{\i}{ {\scriptstyle\sqrt{-1}}\, }
\newcommand{\oom}{ e^{{2\pi \i}/{(n+1)}}  }
\newcommand{\oomn}{ e^{{2\pi \i n}/{(n+1)}}  }
\newcommand{\oomi}{ e^{{2\pi \i i}/{(n+1)}}  }

\renewcommand{\t}{\tilde}

\begin{document}     

\title[The $tt^\ast$
equations]{Nonlinear PDE aspects of the $tt^\ast$
equations of Cecotti and Vafa}  

\author{Martin A. Guest and Chang-Shou Lin}      

\date{}   

\maketitle 

\section{Introduction}\label{intro}

The work of S.~Cecotti and C.~Vafa on topological---anti-topological fusion (see section 8 of \cite{CeVa91}, and also \cite{CeVa92},\cite{CeVa92a}) has pointed the way to some \ll magical solutions\rr of certain systems of partial differential equations.  The main examples appearing in  \cite{CeVa91} are  relatives of the well known two-dimensional Toda lattice
\[
\tfrac{\b^2}{\b z\b \zbar}w_i
=e^{w_{i+1}-w_{i}} - e^{w_{i}-w_{i-1}}
\]
where each $w_i=w_i(z,\zbar)$ is a 
real function of $z\in\C$.
The magical solutions of these equations are predicted by physical results and conjectures. In this article we shall study them from a  mathematical point of view, in order to isolate their essential properties.  In particular, we identify a specific class of $tt^\ast$ equations which includes the equations of Cecotti and Vafa, and we prove an existence/uniqueness result for solutions of some of these equations.  This gives new constructions of \ll global\rr $tt^\ast$ structures, in particular for the orbifold quantum cohomology of several weighted projective spaces and Landau-Ginzburg models.

The Toda lattice itself has various interpretations, e.g.\  in classical field theory (see \cite{YiYa01}) as an example of a nonabelian Chern-Simons theory, and in differential geometry (see
\cite{BoPeWo95} and \cite{BuPe94}) as the equation for primitive harmonic maps taking values in a compact flag manifold.  However, the versions of the Toda lattice which appear in the work of  Cecotti and Vafa are special cases of
\[
\tfrac{\b^2}{\b z\b \zbar}w_i
=-e^{w_{i+1}-w_{i}} + e^{w_{i}-w_{i-1}}
\]
which is the \ll Toda lattice with opposite sign\rrr.  This leads to noncompact Lie groups and solutions with rather different analytic properties.

Our first result (see Definition \ref{stoda} and Proposition \ref{S}) is the description of a class of \ll Toda-like\rr integrable systems which we call the {\em\sstoda.}
The mathematical context for this
is provided by $tt^\ast$ geometry 
(\cite{CeVa91},\cite{Du93},\cite{He03}), a generalization of special geometry. The \stoda has two different interpretations, which generalize the A and B sides of mirror symmetry.  In the language of differential geometry, these are, respectively,
(pluri)harmonic maps with values in the noncompact real symmetric space $\GLNR/\O_n$ (see \cite{Du93}), and 
(pluri)harmonic maps with values in the classifying space of variations of polarized Hodge structures. These Hodge structures can be
finite or infinite-dimensional --- see chapter 10 of \cite{Gu08} for an introduction and references to the well known finite-dimensional version, and 
\cite{Bar01}, \cite{He03}, \cite{KaKoPa07} for the much more recent infinite-dimensional  version.  Our results apply to this infinite-dimensional  version.
 
In this article we shall focus on a simple case for which results on the \ll magical solutions\rr  were not previously known.  This is the case involving two unknown functions (Corollary \ref{twofunctions}), of which the system
\begin{align*}
\tfrac{\b^2}{\b z\b \zbar}w_0&= \ e^{2w_0} - e^{w_1-w_0} 
\\
\tfrac{\b^2}{\b z\b \zbar}w_1&= \ e^{w_1-w_0} - e^{-2w_1}
\end{align*}
is a typical representative.  
Our main technical result (Theorem \ref{mainresult}) is a proof 
using nonlinear p.d.e.\ methods
of the existence and uniqueness of a two-parameter family of solutions parametrized by asymptotic boundary conditions.
For the system above the statement is that, for any parameters 
$\ga_0,\ga_1$ such that
\begin{equation}\label{range}
0\le\ga_0\le2+\ga_1,\ \ 0\le\ga_1\le1,
\end{equation}
there exists a unique solution which satisfies
the conditions
\begin{align*}
&w_i(z) = (\ga_i+o(1)) \log \vert z\vert\
\text{as}\ \vert z\vert\to 0
\\
&w_i(z) \to 0\ 
\text{as}\ \vert z\vert\to \infty.
\end{align*}
It is of interest to note that our method applies to the \stoda but not (directly, at least) to  the Toda lattice itself or other obvious modifications of it. 

This family includes some of the field-theoretic solutions studied by Cecotti and Vafa (so we are able to confirm their predictions for these examples). 
In the case of two unknown functions which we are considering here, the field-theoretic solutions are given by a finite number of Landau-Ginzburg models (unfoldings of certain singularities) and sigma-models (quantum cohomology of certain spaces) corresponding to a finite number of special values of $\ga_0,\ga_1$.    The relation between our solutions and the field-theoretic solutions depends on the well known fact that harmonic maps into symmetric spaces may be constructed from \ll holomorphic data\rr (together with the conjugate \ll anti-holomorphic data\rrr). This is the mathematical manifestation of topological---anti-topological fusion. In quantum cohomology and the theory of Frobenius manifolds it is this holomorphic data which appears explicitly, whereas the harmonic map (or solution to the Toda lattice) is somewhat hidden.  

For both the usual and the opposite sign Toda lattice,  the holomorphic data is a matrix of the form
\begin{equation}\label{data}
\eta=
\bp
 & & & p_0\\
 p_1 & & & \\
  & \ddots & & \\
   & & p_n &
   \ep
\end{equation}
where each $p_i=p_i(z)$ is a holomorphic function.  The open Toda lattice is the special case where at least one $p_i$ is identically zero. The anti-holomorphic data is just given by
$\bar p_0,\dots,\bar p_n$. 

Now, the solutions $w_0,w_1$ in our 
Theorem \ref{mainresult} are radially-invariant, and it follows from this that the corresponding holomorphic data is of the form
\[
p_i(z)=c_i z^{k_i}
\]
for some constants $c_i, k_i$.  The coefficients $\gamma_0,\gamma_1$ of $\log \vert z\vert$ in the asymptotic data can be expressed in terms of $k_0,\dots,k_n$.  Thus, there is a \ll good\rr region in $(k_0,\dots,k_n)$-space for which the conditions 
(\ref{range}) are satisfied.
We shall deduce  (Corollary \ref{conclusion}) that, whenever $k_0,\dots,k_n$ are in this good region,  there exists a solution $w_0,w_1$ associated to the  holomorphic data $c_0z^{k_0},\dots,c_nz^{k_n}$ which is defined on the whole 
of $\C\setminus\{0\}$.  Remarkably, all except one
of the relevant field-theoretic examples are in the good region.   
Thus, these examples can be said to possess \ll global\rr  $tt^\ast$ structures (on $\C\setminus\{0\}$).  

Apart from the fact that we are able to give relatively elementary proofs of the existence of these global $tt^\ast$ structures, two aspects of our method deserve further comment.

First, while the \ll monotone iteration scheme\rr that we shall use is a well known tool for solving certain kinds of nonlinear {\em scalar} p.d.e., it does not generally apply to {\em systems.} The particular combinations of exponential functions which occur in the \stoda are crucial for its applicability in our situation.   Moreover, while uniqueness results of the above type cannot be expected for general systems, for the \stoda we are able to use the maximality property of our solutions together with certain Pohozaev identities to obtain a uniqueness result.  The somewhat surprising effectiveness of these methods is evidence for the special nature of the \sstoda.

Second, although our solutions are all radially-invariant and hence may be regarded merely as solutions of two coupled ordinary differential equations of Painlev\'e type, in general one expects solutions to have many singularities. It is of interest to consider the geometrical meaning of these singularities (and their complete absence in the case of our solutions). In fact, for the usual Toda lattice, it is easy to produce solutions defined on $\C\setminus\{0\}$. For the Toda lattice with opposite sign (and in particular for the \sstoda) it is not. The reason for this --- the difference between the Iwasawa decompositions for compact and noncompact Lie groups --- is explained in section \ref{tt}.  From the viewpoint of the theory of harmonic maps, our solutions correspond to harmonic maps whose extended solutions remain entirely within a single Iwasawa cell. In contrast, \ll most\rr extended solutions are not confined to a single cell, and the singularities arise when cells are crossed. Thus, we believe our solutions are also of interest in harmonic map theory.

In a future publication we hope to treat the general case of three or more unknown functions. However, it seemed worthwhile to present the simplest case of two functions here with a minimum of technicalities.
The case of one unknown function was already studied by Cecotti and Vafa, and here there are two possibilities, both involving well known equations for a scalar function $w=w(z,\zbar)$.  The first is the elliptic sinh-Gordon equation \ll with positive sign\rrr, 
\[
w_{z\zbar}=\sinh w
\]
which reduces (in the radially-invariant situation) to the third Painlev\'e equation.  This has a distinguished family of smooth solutions on $(0,\infty)$  parametrized by asymptotic conditions at $0$ and $\infty$.  The existence of this family is highly nontrivial, but it follows from extensive work on the third Painlev\'e equation in \cite{MTW77} or \cite{FIKN06}.  
One of these solutions represents the quantum cohomology of $\C P^1$ (see \cite{IrXX},\cite{DoGuRoXX}), and one represents an unfolding of the $A_1$ singularity. The second example is the Tzitzeica equation 
\[
w_{z\zbar}=e^{w}-e^{-2w}
\]
which also reduces to the third Painlev\'e equation. 
This has a family of smooth solutions,  one associated to the quantum cohomology of $\C P^2$ and one associated to the
$A_2$ singularity.  Finally, there is another solution of the Tzitzeica equation, which postdates the work of Cecotti and Vafa, associated to the orbifold (Chen-Ruan) quantum cohomology of the weighted projective space $\P(1,2)$. 
Our method applies also to these examples and provides straightforward proofs of the smoothness of the solutions.

We present our results in the following order.  First, the \stoda is introduced in section \ref{toda}.  The existence and uniqueness theorem for the case of two unknown functions is proved in section \ref{solutions}. In section \ref{tt} we give the holomorphic data for these solutions, and explain the relation with the field-theoretic solutions. The appendix
(section \ref{appendix}) reviews the correspondence between solutions of the Toda lattice and holomorphic data.

The first author is grateful to the JSPS and to the Taida Institute for Mathematical Sciences for financial support.

\section{Toda lattices in $tt^\ast$-geometry}\label{toda}

To establish notation, let us review the usual 
 two-dimensional Toda lattice, which we write in this section in the form
\begin{equation}\label{2DTLw}
 2(w_i)_{\zzb}=e^{2(w_{i+1}-w_{i})} - e^{2(w_{i}-w_{i-1})}
\end{equation}
where the real-valued functions $w_i$ ($i\in\Z$) are defined on some open subset $U$ of $\C=\R^2$.  (The $2w_i$ is convenient here, but in the next section we shall replace it by $w_i$.) We shall be concerned mainly with the periodic Toda lattice
 (of period $n+1$), which is the case where $w_i=w_{i+n+1}$ for all $i$ and  $w_0+\cdots+w_n=0$. 
 
This periodic Toda lattice is 
 known to be integrable in the following sense: 
 
 ---the system of equations can be expressed in \ll zero curvature form\rr $d\om+\om\wedge\om=0$, and, as a consequence of the specific form of this $\om$, 
 
 ---each solution $w_0,\dots,w_n$ of the periodic Toda lattice corresponds, locally, to an ordered set of holomorphic functions
 $p_0,\dots,p_n$.   

\no There is no restriction on $p_0,\dots,p_n$, other than being holomorphic, so this is a very satisfactory result.  It extends the well known formula for the general solution of the Liouville equation in terms of a single holomorphic function, which is a special case of the open Toda lattice.  On the other hand, the formula for
 $w_0,\dots,w_n$ in terms of $p_0,\dots,p_n$ is much more complicated.  Moreover, even in the case of the Liouville equation, the relation between global properties of the solution and those of the corresponding holomorphic data can be subtle (cf.\  \cite{JoWa01},\cite{JoLiWa05}, \cite{LiWa10}, and the theory of minimal surfaces). 
 
Since the argument is spread out over several sources in the literature, 
we give in the appendix  a self-contained proof of the 
construction of $w_0,\dots,w_n$ from $p_0,\dots,p_n$.   In order to explain the equations of Cecotti and Vafa in this section, however,  we just need the form of the matrix-valued $1$-form $\om$ mentioned above.
This is
 \[
 \om=\calA dz + \calB d\zbar
 = (\al+\tfrac1\la\be)dz + (\ga+\la\de)d\zbar,
 \]
where
\[
\al
\!=\!\!
\left(
\begin{array}{c|c|c|c}
\vphantom{(w_0)_{(w_0)}^{(w_0)}}
\!\! {\scriptstyle (w_0)_z}  & & & \\
 \hline
\vphantom{(w_0)_{(w_0)}^{(w_0)}}
  & {\scriptstyle (w_1)_z} & & \\
   \hline
\vphantom{(w_0)_{(w_0)}^{(w_0)}}
  & & \ddots & \\
   \hline
\vphantom{(w_0)_{(w_0)}^{(w_0)}}
  & & & {\scriptstyle (w_n)_z}  \!
\end{array}
\right),
\ \ 
\be
\!=\!\!
\left(
\begin{array}{c|c|c|c}
\vphantom{(w_0)_{(w_0)}^{(w_0)}}
 & & & \! e^{w_0\!-\!w_n} \!\! \\
 \hline
\vphantom{(w_0)_{(w_0)}^{(w_0)}}
\!\! e^{w_1\!-\!w_0}\!  &  & & \\
   \hline
\vphantom{(w_0)_{(w_0)}^{(w_0)}}
  &\  \ddots\  &  & \\
   \hline
\vphantom{(w_0)_{(w_0)}^{(w_0)}}
  & & \! e^{w_n\!-\!w_{n\!-\!1}}\!\! &  \!
\end{array}
\right)
\]
and $\ga=-\al^\ast$, $\de=-\be^\ast$.

The zero curvature equation $d\om+\om\wedge\om=0$ is equivalent to $\calA_{\zbar}-\calB_z=[\calA,\calB]$, and this is equivalent to the system (\ref{2DTLw}) (the coefficients $2$ arise from this calculation).  

The parameter $\la\in S^1$ is called the spectral parameter.  When $n\ge 2$ it makes no difference to the zero curvature equation if $\la$ is set equal to $1$; however, we need $\la$ as it plays an important role in {\em solving} the system (see the appendix).  The starting point for this is the observation that $\om$ is a $1$-form with values in the loop algebra 
\[
\La \slpc=\{ f:S^1\to \slpc \st f \ \text{is smooth,} \}
\]
which is the Lie algebra of the loop group
\[
\La \SLPC=\{ f:S^1\to \SLPC \st f \ \text{is smooth} \}.
\]
Let $\tau:\SLPC\to\SLPC$ be the automorphism
\[
\tau(X)=d_{n+1}^{-1} X d_{n+1}
\]
where
\[
d_{n+1}=\diag(1,\oom,\dots,\oomn);
\]
this induces an automorphism of $\slpc$ given by the same formula. The $\tau$-twisted loop group 
$(\La \SLPC)_\tau$  and loop algebra $(\La \slpc)_\tau$
are defined by imposing the condition
\[
\tau(f(\la))=f(\oom \la)
\]
on loops $f$.  This condition means that the coefficient of $\la^i$ in the Fourier expansion of $f$ lies in the $\oomi$-eigenspace 
$\g_i$ of $\tau$. 
The $1$-form $\om$ actually takes values in 
$\tfrac1\la\g_{-1} + \g_0 + \la\g_1$, hence in
the $\tau$-twisted  loop algebra.
Furthermore, it takes values in the real subalgebra $(\La \su_{n+1})_\tau$, which is the fixed point set of the \ll conjugation\rr map
\[
f=\smallsum \la^i X_i \mapsto - \smallsum \la^{-i} X_i^\ast
\]
on $(\La \slpc)_\tau$. This conjugation map  is 
induced from $c:X\mapsto -X^\ast$ on $\slpc$, which defines the real form
$\Fix(c) = \su_{n+1}$. The corresponding Lie group involution
$C:X\mapsto (X^\ast)^{-1}$ defines the real form $\Fix(C) = \SU_{n+1}$ of $\SLPC$.

The point of these Lie-theoretic remarks is that, not only 
does $\om$ takes values in the \ll real\rr part of $\tfrac1\la\g_{-1} + \g_0 + \la\g_1$, but also the converse statement is true in the sense that any such $\om$ can be transformed to the above form for some functions $w_0,\dots,w_n$ (see the appendix).  Thus, the Toda lattice has a purely Lie-theoretic description. This depends only on having a real form $G = \Fix(C)$ of a complex Lie group $G^\C$ and an automorphism $\tau$. 

A better-known Lie-theoretic description is that, in terms of the variables
\[
u_i=2w_i-2w_{i-1},
\]
equation (\ref{2DTLw}) becomes
\begin{equation}\label{2DTLu}
(u_i)_{z\zbar} = e^{u_{i+1}} - 2e^{u_{i}}+e^{u_{i-1}}.
\end{equation}
i.e.\
\[
 (u_i)_{\zzb} = -\smallsum_{j=0}^{n} \,k_{ij} e^{u_j}
\]
where $(k_{ij})_{0\le i,j\le n}$ is the Cartan matrix of 
$\La \slpc$.  The automorphism $\tau$ and the involution $C$ giving the (compact) real form are both determined intrinsically by the Cartan matrix.  Clearly, this allows one to generalize the Toda lattice to other Lie algebras or affine Lie algebras (or, more generally, root systems).  For details of such \ll Toda-type systems\rr we refer to  \cite{RaSa97},\cite{RaSa97a}.  
Affine Lie algebras include the loop algebras 
$\La\g^{\C}$ and also the twisted loop algebras
\[
(\La \g^{\C})_\th = \{ f\in \La \g^{\C} \st 
\th(f(\la))=f(
e^{{2\pi \i}/{N}}
\la) \}
\]
where $\th$ is an automorphism of $\g^{\C}$ of order $N$.  
If $\th_1,\th_2,\dots$ are automorphisms, the notation
$(\La \g^{\C})_{\th_1,\th_2,\dots}$ means
$(\La \g^{\C})_{\th_1}\cap
(\La \g^{\C})_{\th_2}\cap
\cdots.
$

In this article we have in mind a different generalization. We fix $\slpc$ and $\tau$, but we allow various real forms and (compatible) involutions $\si$.
A \ll Toda-like\rr system means {\em a system of equations given by any real form of} $(\La \slpc)_\tau$ or $(\La \slpc)_{\tau,\si}$ {\em such that the conjugation map preserves} $\g_0$ {\em and interchanges $\g_{-1}$ with $\g_1$.}  

For example, the real form of $(\La \slpc)_\tau$ given by the conjugation map
\[
\smallsum \la^i X_i \mapsto - \smallsum (-1)^i \la^{-i} X_i^\ast
\]
produces the \ll Toda lattice with opposite sign\rrr, namely
\begin{equation}\label{oppositesign2DTL}
 2(w_i)_{\zzb}=-e^{2(w_{i+1}-w_{i})} +e^{2(w_{i}-w_{i-1})}
\end{equation}
(or
$(u_i)_{\zzb} = \smallsum_{j=0}^{n} \,k_{ij} e^{u_j}$
in terms of the variables
$u_i=2w_i-2w_{i-1}$).  This appears prominently in the work of Cecotti and Vafa, though always with the additional symmetry
\begin{equation}\label{antisymm}
w_i + w_{n-i}=0,
\end{equation}
which is equivalent to imposing the additional twisting condition
$\si(f(\la))=f(-\la)$ where
\[
\si(X) = 
-\De
X^t
\De,
\quad
\De=
\bp
 & & \,1\\
 & \iddots & \\
1 & &
\ep.
\]
It turns out that the system given by equations (\ref{oppositesign2DTL}) and 
(\ref{antisymm}) is the case $S=N=\De$ of the following family of examples:

\begin{definition}\label{stoda}  Let $\S$ be a symmetric nondegenerate complex $(n+1)\times(n+1)$-matrix.  Let $P$ be any matrix such that $\S=(P^t)^{-1}P^{-1}$. Let $N=P\bar P^{-1}$. We define a conjugation map
\[
c(X)=N\bar X N^{-1}
\]
and an involution
\[
\si(X)=-\S^{-1}X^t\S
\]
on $\slpc$ (it follows that $c$ and $\si$ commute).  If 
$c$ and $\si$ commute with 
\[
\tau(X)=d_{n+1}^{-1} X d_{n+1}
\]
then the resulting Toda-like system will be called the {\em \sstoda.}
\end{definition}

To explain this definition, we must introduce some notation. First, we interpret 
$\S$ as the matrix of a nondegenerate symmetric bilinear form
\[
\lan x,y\ran_\S = x^t \S y
\]
on $\C^{n+1}$ (where $x,y$ are column vectors with respect to the standard basis $e_0,\dots,e_n$).  Since $\S$ is nondegenerate, there exists a basis $Pe_0,\dots,Pe_n$ of $\C^{n+1}$ with respect to which the matrix of $\lan \ ,\ \ran_\S$ is the identity matrix.  Hence 
$\S=(P^t)^{-1}P^{-1}$. 
The complex orthogonal group (with respect to 
$\lan \ ,\ \ran_\S$)
\begin{align*}
\SO^\S_{n+1}\C&=
\{ X\in\SL_{n+1}\C \st 
\lan Xx,Yy\ran_\S=\lan x,y\ran_\S \}\\
&=P\, \SO_{n+1}\C\, P^{-1}
\end{align*}
can be described as the fixed point set of the involution $X\mapsto \S^{-1}(X^t)^{-1}\S$ of $\SL_{n+1}\C$. This induces the involution $\si(X)=-\S^{-1}X^t \S$ of $\sl_{n+1}\C$.

The real subspace $P\R^{n+1}$ is the fixed point set of 
the $\R$-linear involution
\[
B(x)=N\bar x,
\]
where $N=P\bar P^{-1}$.
Using this, we obtain the real form
\begin{align*}
\SL^N_{n+1}\R&=\{ X\in\SL_{n+1}\C \st X \Fix(B)\sub \Fix(B)\}\\
&= P \, \SL_{n+1}\R \, P^{-1}
\end{align*}
of $\SL^N_{n+1}\C$,
which can also be described as the fixed point set of the conjugation map $C(X)=N\bar X N^{-1}$.   This induces
$c(X)=N \bar X N^{-1}$ on $\sl_{n+1}\C$.

The restriction of $\lan\ ,\ \ran_\S$ to $\Fix(B)$ is a positive-definite real-valued inner product; in fact for $Px,Py\in P\R^{n+1}$ we have $\lan Px,Py\ran_\S=x^t P^t \S P y = x^ty$.  
We denote the orthogonal group with respect to this inner product by
\[
\SO^{\S,N}_{n+1}= \SL^N_{n+1}\R\ \cap \ \SO^\S_{n+1}\C
= P\, \SO_{n+1} \,P^{-1}.
\]

With this notation, we can explain the dual aspects of
Definition \ref{stoda}, represented by $\tau$ and $\si$, which in turn explains why the \stoda  describes certain examples arising in mirror symmetry.  Namely, if we ignore the 
involution $\si$, then a solution has the standard differential geometric interpretation as a primitive harmonic map to a flag manifold.  A variation of polarized Hodge structure would give such a primitive harmonic map;
this harmonic map exhibits the \ll B-model side of mirror symmetry. On the other hand, if we ignore the automorphism $\tau$, we obtain a quite different kind of harmonic map, namely a harmonic map to the symmetric space 
\begin{align*} 
\SL^N_{n+1}\R/\SO^{\S,N}_{n+1}&\cong
P\, \SL_{n+1}\R\,P^{-1}/P\,\SO_{n+1}\,P^{-1}\\
&\cong
\SL_{n+1}\R/\SO_{n+1}.
\end{align*}
This is the $tt^\ast$ property expected for Frobenius manifolds, as explained in
\cite{Du93}.  It could be described as the \ll A-model side\rr of the \sstoda.

In \ll true\rr mirror symmetry one encounters
the situation that $\lan x,y\ran_\S$ is the intersection form of ordinary cohomology of a manifold, and the real subspace $\Fix(B)$ is the real cohomology of a mirror partner.  Motivated by this, we shall assume that 
\[
\S=
\bp
T_{l_1}\De_{l_1} & & & \\
 & \!\!\!\!\!\!\!\!  T_{l_2}\De_{l_2} & & \\
  & & \ddots & \\
 & & &  T_{l_r}\De_{l_r}
\ep,
\quad
\De_l = (\de_{i,l-j})_{1\le i,j\le l} =
\bp
 & & \,1\\
 & \iddots & \\
1 & &
\ep
\]
for some diagonal matrices $T_{l_1},\dots,T_{l_r}$ with positive diagonal entries such that $T_{l_i} \De_{l_i}=
\De_{l_i} T_{l_i}$
(in other words, utilizing the equivalence of all complex symmetric nondegenerate bilinear forms, we choose this particular representative as our starting point).  

In the spirit of our definition of the \stoda one could consider any $B$ such that the restriction of $\lan x,y\ran_\S$
to $\Fix(B)$ is positive definite. However, this does not lead to a more general definition than Definition \ref{stoda}.  
In fact, in terms of the above normalization of $S$, we can reduce the possibilities still further:

\begin{proposition}\label{S}  Consider $c$, $\si$, $\tau$ as in Definition \ref{stoda}, with $\S$ written in the above form.

(1) There exists a matrix $P$ such that 
$\S=(P^t)^{-1}P^{-1}$
and
\[
N=P\bar P^{-1}=
\bp
\De_{l_1} & & & \\
 & \!\!\!\! \De_{l_2} & & \\
  & & \ddots & \\
 & & &  \De_{l_r}
\ep.
\]
 
(2) The condition that $c$ and $\si$ commute with 
$\tau$ forces $r=1$ or $2$.
\end{proposition}

\begin{proof} (1) It suffices to prove this in the case $r=1$. Thus, we need a matrix $P$ such that $(P^t)^{-1}P^{-1}=T\De$ and
$P\bar P^{-1}=\De$, where 
$T=\diag(t_0,\dots,t_n)=\diag(t_n,\dots,t_0)$ and
all $t_i>0$. We claim that $P=T^{-\frac12}\sqrt{-i}C$ satisfies these conditions, where
\newcommand{\qmaspace}{\hspace{.08cm}}
\newcommand{\andq}{\qmaspace & \qmaspace}
\[
C=\frac{\ 1}{\sqrt 2}
\bp
1 \andq \andq \andq i\\
 \andq \qmaspace\ddots\qmaspace \andq \qmaspace\iddots\qmaspace \andq \\
  \andq \iddots \andq \ddots \andq \\
i \andq \andq \andq 1
\ep.
\]
This follows from the fact that $TC=CT$ and
$C^2=i\De$, $\bar C=C^{-1}$, $C^t=C$.  Namely,
$(P^t)^{-1}P^{-1}=P^{-2}=TiC^{-2}=T\De$,
and $P\bar P^{-1}=T^{-\frac12}\sqrt{-i}CT^{\frac12}\sqrt{-i}C=(-i)i\De=\De$.

(2) We have $c\circ\tau=\tau\circ c$ if and only if 
$d_{n+1}Nd_{n+1}^{-1}$ is a scalar multiple of $N$. This holds if $r=1$ or $2$, but not if $r\ge 3$. A similar argument applies to $\si$.
\end{proof}

If all the $T_{l_i}$ are identity matrices, then we obtain the
\ll Toda lattice with opposite sign\rr with the following additional conditions:

\no $r=1$: \ \  $w_i+w_{n-i}=0$ for $0\le i\le n$ (these are the
equations of Cecotti and Vafa);

\no $r=2$: \ \  $w_i+w_{l_1-i-1}=0$ for $0\le i\le l_1-1$ and $w_i+w_{n+l_1-i}=0$ for $l_1\le i\le n$, with $l_1>1$ or $l_2>1$).

\no For general $T_{l_i}$, the equations of the \stoda can still be reduced to one of these two forms (see the appendix).  

\begin{corollary}\label{twofunctions} Any system arising from the \stoda which involves two unknown functions can be written in the form
\begin{equation}\label{generalsystem}
\begin{cases}
w_{z\zbar}&= \ e^{aw} - e^{v-w} 
\\
v_{z\zbar}&= \ e^{v-w} - e^{-bv}
\end{cases}
\end{equation}
where $a,b\in\{1,2\}$.
\end{corollary}

\begin{proof} This follows from a case by case analysis, which we summarize in the first three columns of Table 1. There are ten possibilities for $(l_1,l_2)$.  With the indicated choices for $w,v$ we obtain four possibilities for $(a,b)$, as asserted. 
For later convenience we give the form of the holomorphic data $p_i$ in the fourth column, and the relations between the functions $h_i$ (see part (i) of section \ref{tt}) in the last column. The symbol $[ij\dots]$ in this column means that $h_ih_j\dots=1$. These conditions on $p_i$ and $h_i$ (respectively) follow directly from the definitions of
$(\La\sl_{n+1}\C)_\si$ and $(\La\SL_{n+1}\C)_\si$.
\end{proof}  

\begin{table}[ht]
\renewcommand{\arraystretch}{1.3}
\begin{tabular}{cc|cc|cc||c|c}
$l_1$ & $l_2$ & $w$ & $v$ & $a$ & $b$
& $p_0,\dots,p_n$ & $h_0,\dots,h_n$
 \\
\hline
$4$ &  & $2w_0$ & $2w_1$ & 2 & 2 & $p_0,p_1\equ p_3,p_2$ & $[03],[12]$
\\
$5$ &  & $2w_0$ & $2w_1$ & 2 & 1 & 
$p_0,p_1\equ p_4,p_2\equ p_3$&
$[04],[13],[2]$
\\
$1$ & $4$ & $2w_1$ & $2w_2$ & 1 & 2 & $p_0\equ p_1,p_2\equ p_4,p_3$ & $[0],[14],[23]$
\\
$1$ & $5$ & $2w_1$ & $2w_2$ & 1 & 1& $p_0\equ p_1,p_2\equ p_5,p_3\equ p_4$ &
$[0],[15],[24],[3]$
\\
$2$ & $2$ & $2w_3$ & $2w_0$ & 2 & 2& 
$p_0\equ p_2,p_1,p_3$ & $[01],[23]$
\\
$2$ & $3$ & $2w_4$ & $2w_0$ & 1 & 2 & 
$p_0\equ p_2,p_1,p_3\equ p_4$ &
$[01],[24],[3]$
\\
$3$ & $2$ & $2w_4$ & $2w_0$ & 2 & 1& 
$p_0\equ p_3,p_1\equ p_2,p_4$ &
$[02],[1],[34]$
\\
$3$ & $3$ & $2w_5$ & $2w_0$ & 1 & 1 &
$p_0\equ p_3,p_1\equ p_2,p_4\equ p_5$  &
$[02],[1],[35],[4]$
\\
$4$ & $1$ & $2w_0$ & $2w_1$ & 1 & 2 & 
$p_0\equ p_4,p_1\equ p_3,p_2$ &
$[03],[12],[4]$
\\
$5$ & $1$ & $2w_0$ & $2w_1$ & 1 & 1 & 
$p_0\equ p_5,p_1\equ p_4,p_2\equ p_3$ &
$[04],[13],[2],[5]$
\end{tabular}
\bigskip
\caption{}
\end{table}

\begin{remark}\label{derivatives}  
The particular choices of $w,v$ in Table 1 were made so that, if the equations of the system are written as
$w_{z\zbar}=F(w,v)$, $v_{z\zbar}=G(w,v)$, then
$\tfrac{\b}{\b v} F(w,v) < 0$, $\tfrac{\b}{\b w} G(w,v) < 0$. 
For example, in the case $(l_1,l_2)=(2,2)$, we have
$\tfrac{\b}{\b w_1}
\left( e^{2w_0} - e^{w_1-w_0} \right)
= -e^{w_1-w_0} < 0$, and
$\tfrac{\b}{\b w_0}
\left( e^{w_1-w_0} - e^{-2w_1} \right)
= -e^{w_1-w_0} < 0$.
This property will be essential in the next section.  It is a feature of the \sstoda; in fact, it essentially characterizes the \sstoda, in the sense that $\S=\diag(\De_{l_1},\dots,\De_{l_r})$ has this property if and only if $r=1$ or $r=2$. This confluence of good Lie algebraic properties and good analytic properties is further evidence of the importance of the \sstoda.
\end{remark}

\section{A class of distinguished solutions}\label{solutions}

As holomorphic functions will not play any role in this section, we shall sometimes write $x=(x_0,x_1)\in\R^2$ instead of $z=x_0+ix_1\in \C$, and
$\De = 
4\tfrac{\b^2}{\b z\b \zbar}=
\tfrac{\b^2}{\b x_0^2} + \tfrac{\b^2}{\b x_1^2}$,
$r=\vert z\vert=\vert x\vert.$
All functions in this section are assumed (or proved to be) smooth on the domain $\R^2\setminus\{(0,0)\}=\C\setminus\{0\}$ unless stated otherwise.  In particular an inequality such as $u< v$ means that 
$u(z)< v(z)$ for all $z\in \C\setminus\{0\}$.

We shall obtain a family of solutions of the system 
\begin{equation*}
\begin{cases}
(w_0)_{z\zbar}&= \ e^{aw_0} - e^{w_1-w_0} 
\\
(w_1)_{z\zbar}&= \ e^{w_1-w_0} - e^{-bw_1}
\end{cases}
\end{equation*}
(system (\ref{generalsystem}) from section \ref{toda}),
where $a,b\in\{1,2\}$.  In fact our proof works for any $a,b>0$.

\begin{theorem}\label{mainresult}
For $a,b>0$, the above system 
has a unique solution $(w_0,w_1)$ which satisfies
the boundary conditions
\begin{equation*}
\begin{cases}
\ w_i(z) = (\ga_i+o(1)) \log \vert z\vert\
\text{as}\ \vert z\vert\to 0
\\
\ w_i(z) \to 0\ 
\text{as}\ \vert z\vert\to \infty
\end{cases}
\end{equation*}
for any $(\ga_0,\ga_1)\in\R^2$ such that
\[
0\le\ga_0\le2+\ga_1,\quad
0\le\ga_1\le 2/b.
\]
\end{theorem}  

\begin{remark} 
{\bf (i)\,} The upper bounds on $\ga_0,\ga_1$ are optimal, as no term on the right hand sides of either of the equations can have singular behaviour worse than that of $\vert z\vert^{-2}$, as $z\to 0$. Thus, for any solution, we must have
$\ga_1-\ga_0\ge-2$ and $-b\ga_1\ge-2$.  

\no{\bf (ii)\,} In the \ll interior\rr case $0\le\ga_0<2+\ga_1$, 
$0\le\ga_1< 2/b$, our proof shows that the stronger boundary condition
$w_i(z)=\ga_i\log\vert z\vert + O(1)$ holds as $z\to 0$.

\no{\bf (iii)\,} It is easy to see that our proof works also when 
$-2/a\le\ga_0\le 0$, $-2+\ga_0\le\ga_1\le 0$. 

\no{\bf (iv)\,} We shall give the proof for the case $a=b=2$. Therefore,  for the remainder of the section, we consider the system
\begin{equation}\label{system}
\begin{cases}
(w_0)_{z\zbar}&= \ e^{2w_0} - e^{w_1-w_0} 
\\
(w_1)_{z\zbar}&= \ e^{w_1-w_0} - e^{-2w_1}
\end{cases}
\end{equation}
subject to the boundary conditions
\begin{equation}\label{boundary}
\begin{cases}
\ w_i(z) = (\ga_i+o(1)) \log \vert z\vert\
\text{as}\ \vert z\vert\to 0
\\
\ w_i(z) \to 0\ 
\text{as}\ \vert z\vert\to \infty
\end{cases}
\end{equation}
with $0\le\ga_0\le2+\ga_1$,
$0\le\ga_1\le1$.
The other cases may be treated in exactly the same way.

\no{\bf (v)\,} The proof will use (a) {\em a priori} upper and lower bounds on solutions, (b) an iteration procedure to prove existence of (maximal) solutions, and (c) certain integral identities to prove uniqueness. Before starting the proof, we summarize these ingredients briefly.  (a) An elementary argument (Proposition \ref{super}) shows that any solution of (\ref{system}), (\ref{boundary})
satisfies $w_0\le 0, w_1\le 0$.  Then (Proposition \ref{sub}) we shall find $q_0,q_1$ such that $w_0\ge q_0, w_1\ge q_1$.  
To establish the existence of $q_0,q_1$ we need Lemmas \ref{h1}, \ref{h2}, and \ref{q}. (b) Next, we shall produce monotone sequences
\[
q_i \le \dots \le  w_i^{(n+1)} \le w_i^{(n)} \le \dots \le w_i^{(0)} \le 0
\]
whose limits $w_i=\lim_{n\to\infty} w_i^{(n)}$
are (maximal) solutions of (\ref{system}), (\ref{boundary}), thus establishing existence.  Our argument will make use of the precise form of the coefficients of the exponentials in the system (see Remark \ref{concluding} at the end of the proof).
(c) Finally, to prove uniqueness of these solutions, we derive Pohozaev-type identities which relate $\ga_0,\ga_1$ to certain integrals of the solutions. 
\end{remark}
 
Let us begin by establishing upper and lower bounds on solutions of (\ref{system}), (\ref{boundary}).

\begin{proposition}\label{super} Any solution  of (\ref{system}), (\ref{boundary})
 satisfies $w_0\le 0, w_1\le 0$.
\end{proposition}

\begin{proof}  We prove this by contradiction. Let us suppose that $w_0$ is positive somewhere in $\C\setminus\{0\}$. 
The boundary conditions (\ref{boundary}) imply that $w_0$ takes a maximum value, say at $z_0\in \C\setminus\{0\}$, hence $(w_0)_{z\zbar}(z_0)\le 0$.
Then  $e^{2w_0} - e^{w_1-w_0} = (w_0)_{z\zbar} \le 0$ at $z_0$, hence
$w_1(z_0)-w_0(z_0)\ge 2w_0(z_0)$, so we have
$w_1(z_0)\ge 3w_0(z_0)>0.$
From the boundary conditions, $w_1$ also takes a maximum value, say at $z_1$. Again
$(w_1)_{z\zbar}(z_1) \le 0$ implies
$w_0(z_1) \ge 3w_1(z_1) \ge 3w_1(z_0) 
\ge 9w_0(z_0)$, which contradicts 
$w_0(z_0)>0$.
It follows that $w_0\le 0$.  Similarly, $w_1\le 0$.
\end{proof}

It is more difficult to establish lower bounds.  For this purpose, we consider first the following scalar equation:
 
\begin{lemma}\label{h1} Let $\ga\ge0$. Then the equation
$h_{z\zbar} = e^{2h} - 1$ has a unique solution which satisfies the boundary conditions 
$h(z)=\ga\log\vert z\vert + O(1)$ 
as $\vert z\vert\to 0$, 
$h(z) \to 0$ as $\vert z\vert\to \infty$.
\end{lemma}
 
This is well known (see sections III.3 and III.4 of  \cite{JaTa80},\cite{YiYa01}), so we omit the proof.  In fact it can also be proved by a monotone iteration scheme similar to, but easier than, the one we shall use to solve (\ref{system}), (\ref{boundary}).

\begin{lemma}\label{h2} The function $h$ of Lemma \ref{h1}, with $\ga=\ga_0+\ga_1$,  satisfies
$h\le w_0+w_1$.
\end{lemma}

\begin{proof} The function $h$ of  Lemma \ref{h1} depends continuously on $\ga$; let $h_\eps$ denote the solution given by $\ga_\eps=\ga+\eps$. It will suffice to prove that $h_\eps\le w_0+w_1$ for any $\eps>0$, as we obtain $h\le w_0+w_1$ by taking the limit $\eps\downarrow 0$. For this we shall use two facts:

(a) From the system (\ref{system}) we have
\[
(w_0+w_1)_{z\zbar} = e^{2w_0} - e^{-2w_1}
= e^{-2w_1}(e^{2w_0+2w_1} - 1)
\le e^{2w_0+2w_1} - 1
\]
(here we use the fact that $w_0,w_1\le 0$). Hence
$(w_0+w_1-h_\eps)_{z\zbar} \le 
e^{2(w_0+w_1)} - e^{2h_\eps}$.

(b) The boundary conditions on $h_\eps$ and $w_0,w_1$ show that, for any $\eps>0$, if 
$\inf(w_0+w_1 - h_\eps)<0$, then 
$\inf(w_0+w_1 - h_\eps)$ is assumed at some point $z_0$, in which case we have $(w_0+w_1-h_\eps)_{z\zbar}(z_0)\ge0$. 

Now, if it is false that $h_\eps\le w_0+w_1$, 
then $w_0(z_0)+w_1(z_0)-h_\eps(z_0)< 0$.  By (a)
we have $(w_0+w_1-h_\eps)_{z\zbar}(z_0)<0$.  This contradicts (b). Thus $h_\eps\le w_0+w_1$, hence also $h\le w_0+w_1$.
\end{proof}

We shall make use of the above maximum principle argument\footnote{That is, to prove by contradiction an inequality of the form $f\ge0$, we prove an estimate of the form $\De f\le F(f)$ and simultaneously show that $f$ takes a local minimum.  
Since $\De f\ge 0$ at a local minimum, and also $f<0$ by assumption, we obtain a contradiction if the estimate can be used to show that $\De f< 0$ (for example, if $F(f)$ is a positive function times $f$).  Another application of this method is to prove uniqueness of solutions to an equation of the form $\De g=G(g)$: take $f=g_1-g_2$ and then $f=g_2-g_1$, where $g_1,g_2$ are any two solutions satisfying appropriate boundary conditions.} repeatedly. As the
details are all very similar we omit them from now on.

\begin{lemma}\label{q}   Let  $0\le\ga_0\le 2+\ga_1$,
$0\le \ga_1\le 1$, and let $h$ be as in Lemma \ref{h2} with 
$\ga=\ga_0+\ga_1$.  Then:

(1) The equation
$(q_0)_{z\zbar} = e^{2q_0} - e^{h-2q_0}$ has a unique solution which satisfies the boundary conditions 
$q_0(z)=(\ga_0+o(1))\log\vert z\vert$ 
as $\vert z\vert\to 0$, 
$q_0(z) \to 0$ as $\vert z\vert\to \infty$.

(2) The equation
$(q_1)_{z\zbar} = e^{2q_1-h} - e^{-2q_1}$ has a unique solution which satisfies the boundary conditions 
$q_1(z)=(\ga_1+o(1))\log\vert z\vert$ 
as $\vert z\vert\to 0$, 
$q_1(z) \to 0$ as $\vert z\vert\to \infty$.

(3) The functions $q_0,q_1$ obtained in (1),(2) satisfy 
$h\le q_0+q_1$ and $q_0,q_1\le 0$.
\end{lemma}

\begin{proof} (1) For any $\eps>0$, let $f^\eps\in C^\infty(\R^2)$ be nonnegative, radially-invariant, decreasing with respect to $\vert x\vert$, such that $f^\eps$ has support in the unit disk $B_1=\{ x\in\R^2 \st \vert x\vert\le 1\}$ and converges weakly as $\eps\downarrow 0$ to $\frac\pi2 \ga_0\de_0$, where $\de_0$ is the Dirac measure at $0$.  Define $h^\eps$ by
\[
h^\eps(x) =
\begin{cases}
h(x)\quad \vert x\vert \ge \eps \\
h(\eps)\quad\, \vert x\vert < \eps 
\end{cases}
\]
where $h$ is the function of Lemma \ref{h1}, with $\ga=\ga_0+\ga_1$.

{\em We claim that, for any $\eps>0$, the equation
\[
(q^\eps)_{z\zbar} = e^{2q^\eps} - e^{h^\eps-2q^\eps}
+ f^\eps,\quad
q^\eps:\R^2\to\R
\]
has a unique solution $q^\eps$ such that
$q^\eps(x) \to 0$ as $\vert x\vert\to \infty$.}

Uniqueness is clear, by the maximum principle.  In particular, it follows that a solution (if it exists) must be radially-invariant. 

To prove existence of $q^\eps$, we begin by considering the equation
\[
(q^{\eps,R})_{z\zbar} = e^{2q^{\eps,R}} - e^{h^\eps-2q^{\eps,R}}
+ f^\eps,\quad
q^{\eps,R}:B_R\to\R
\] 
for $q^{\eps,R}$ on the ball $B_R$ of (large) radius $R$, subject to the boundary condition $q^{\eps,R}\vert_{\b B_R}=0$.  Let
\[
J(v)=\tfrac18 \smallint_{\!B_R}\ \vert \nabla v\vert^2 
+ \tfrac12 \smallint_{\!B_R}\ e^{2v} + e^{h^\eps-2v}
- \smallint_{\!B_R}\  f^\eps v 
\]
for $v\in H^1_0(B_R)=\{
v \st \smallint_{\!B_R}\ \vert \nabla v\vert^2  < \infty
\text{ and } v=0 \text{ on } \b B_R \}$.  Suppose
that $v_i$ is a minimizing sequence for $J$, i.e.\
$J(v_i)\to \inf J$ as $i\to\infty$. Then we have
$\smallint_{\!B_R}\ \vert \nabla v_i\vert^2  \le C$
for some constant $C$.  Since $L^2$ is compactly embedded in $H^1_0$, there exists a subsequence (still denoted by $v_i$) with the properties $v_i\rightharpoonup v_\infty$ 
(i.e.\ converges weakly) in $H^1_0$,
$v_i\to v_\infty$ in $L^2$, and $v_i(x)\to v_\infty(x)$ for almost all $x$ in $\R^2$.  Thus 
\[
\lim_{i\to\infty} \smallint_{\!B_R}\ \vert \nabla v_i\vert^2 
\ge \smallint_{\!B_R}\ \vert \nabla v_\infty\vert^2,
\ \ 
\lim_{i\to\infty} \smallint_{\!B_R}\ f^\eps v_i = 
\smallint_{\!B_R}\ f^\eps v_\infty
\]
and
\[
\lim_{i\to\infty} \smallint_{\!B_R}\ e^{2v_i} + e^{h^\eps - 2v_i}
\ge \smallint_{\!B_R}\ e^{2v_\infty} + e^{h^\eps - 2v_\infty}
\]
by Fatou's Lemma. Thus
\[
J(v_\infty) \le \inf_{v\in H^1_0(B_R)} J(v)
\]
i.e.\ the minimum of $J$ is attained by $v_\infty$. It is easy to see that $v_\infty$ is the required solution $q^{\eps,R}$. By the maximum principle we have $q^{\eps,R}\le 0$ on $B_R$. 

We claim that $\tfrac{d}{dr} q^{\eps,R} \ge 0$ for $r=\vert x\vert \in [0,R]$.  If not, then the set
\[
\Om=\left\{ 
x=(x_0,x_1)\in B_R 
\vphantom{\tfrac{\b}{\b x_0} q^{\eps,R}(x)}
\right.
\left| \ 
\tfrac{\b}{\b x_0} q^{\eps,R}(x)<0
\text{ and }
x_0>0 \right\}
\]
is nonempty.  Set $\phi(x)= \tfrac{\b}{\b x_0} q^{\eps,R}(x)$.  Then $\phi$ satisfies 
\[
\phi_{z\zbar} - 2(e^{2q^{\eps,R}} + e^{h^\eps-2q^{\eps,R}})\phi
= - e^{h^\eps-2q^{\eps,R}}\tfrac{\b h^\eps}{\b x_0} +
\tfrac{\b f^\eps}{\b x_0} \le 0
\]
in $\Om$, because $h^\eps(x)=h^\eps(\vert x\vert)$ is increasing in $\vert x\vert$ and $f^\eps(x)$ decreasing. 
Multiplying both sides by $\phi\,(\le 0)$ and integrating, we obtain
\[
\smallint_\Om  \ \vert \phi_z \vert^2 +
 2(e^{2q^{\eps,R}} + e^{h^\eps-2q^{\eps,R}})\phi^2
 \le 0,
\]
which is a contradiction.  Thus we have proved that 
$\tfrac{d}{dr} q^{\eps,R} \ge 0$.

By applying the maximum principle at $x=0$, which is the minimum of $q^{\eps,R}$, we have
$0\le 
e^{2q^{\eps,R}} - e^{h^\eps-2q^{\eps,R}} + f^\eps$
at $x=0$, which implies 
$q^{\eps,R}(x) \ge q^{\eps,R}(0) \ge -C_\eps$,
for some positive constant $C_\eps$ independent of $R$. 

For $R^\pr > R$,  the maximum principle shows that $q^{\eps,R}(x)\ge 
q^{\eps,R^\pr}(x)$ for $\vert x\vert \le R$.  Thus, by letting $R\to\infty$, we see that $q^{\eps,R}$ converges to some $q^\eps$.  Clearly, $q^\eps$ is increasing in $r$. 

Finally we let $\eps\downarrow 0$.  We claim that $q^\eps$ converges on $\C\setminus\{0\}$.  If not, there exists some $r_0>0$ and a sequence of values $\eps_n\downarrow 0$ such that
$q^{\eps_n}(r_0)\to-\infty$ and hence
$q^{\eps_n}(r)\to-\infty$ for all $r\in [0,r_0]$. Integrating over $B_{r_0}$ for such $\eps=\eps_n$, we obtain
\begin{equation}\label{inequality}
0\le \tfrac{2\pi}4 \tfrac{d q^{\eps}}{d r}(r_0) r_0
=
\smallint_{\!B_{r_0}} \ 
e^{2q^{\eps}} - e^{h^\eps-2q^{\eps}}
\ 
+
\ 
\smallint_{\!B_{r_0}} \ 
f^\eps.
\end{equation}
However, as $\eps\to 0$, we have
\[
\smallint_{\!B_{r_0}} \ 
e^{2q^{\eps}}
\ 
+
\ 
\smallint_{\!B_{r_0}} \ 
f^\eps = O(1)
\]
and
\[
\smallint_{\!B_{r_0}} \ 
e^{h^\eps-2q^{\eps}}
\ \ge\ 
e^{h(r_0/2)}
\smallint_{\, B_{r_0}\,\setminus \,B_{r_0/2}\,} \ 
e^{-2q^{\eps}}
\to\infty,
\]
which contradicts (\ref{inequality}).  Thus,
$q^\eps$ converges to some $q$ on $\C\setminus\{0\}$.
From
\[
\tfrac{d q^{\eps}}{d r} r \ \le\ 
\tfrac2\pi \,\smallint_{\!B_{r}} \ 
f^\eps = \ga_0
\]
we have
$q^\eps(1)-q^\eps(r) \le \ga_0\log\tfrac1r$, hence
$e^{h^\eps(r)-2q^{\eps}(r)} \le C r^{\ga_1-\ga_0}$
for some constant $C$.
If $\ga_0-\ga_1<2$, then $e^{h^\eps(r)-2q^{\eps}(r)}$
is bounded by an $L^1$ function. It is then easy to see that
$q(r) = \ga_0\log r + O(1)$.  This completes the existence  part of the proof when $\ga_0-\ga_1<2$.

If $\ga_0=2+\ga_1$, choose some small $\eps>0$ and
let $q_\eps$ be the solution obtained above for the case 
$\ga_\eps = \ga_0-\eps=2+\ga_1-\eps$. By the maximum principle we have $q_{\eps}>q_{\eps^\pr}$ whenever $\eps>\eps^\pr>0$. Let $q=\lim_{\eps \downarrow 0} q_\eps$.  We have
\begin{align*}
\tfrac2\pi \ 
\smallint_{\!B_r} \ 
e^{h(x)-2q^{\eps}(x)} dx 
&= 
\tfrac2\pi\ 
\smallint_{\!B_r} \ 
e^{2q^{\eps}(x)} dx
+ \ga_\eps - \tfrac{d q^{\eps}}{d r} r
\\
&\le 
\tfrac2\pi\ 
\smallint_{\!B_r} \ 
e^{2q^{\eps}(x)} dx
+ \ga_\eps 
\\
&\le C.
\end{align*}
Since $h(x)-2q^{\eps}(x)$ is monotone in $\eps$, 
the monotone convergence theorem gives
\[
\smallint_{\!B_r} \ 
e^{h-2q} = 
\lim_{\eps\downarrow 0} 
\smallint_{\!B_r} \ 
e^{h-2q^{\eps}}
\le C.
\]
Thus $e^{h-2q} \in L^1(B_r)$ for all $r>0$, and 
\[
\tfrac{d q}{d r} r 
\ = \ 
\tfrac2\pi\ 
\smallint_{\!B_r} \ 
(e^{2q(x)}- e^{h(x)-2q(x)})\  dx
+ \ga_0,
\]
which implies the required result as in the previous case.  This completes the proof of part (1) of Lemma \ref{q}.

The proof of (2) is similar, and (3) is an application of the maximum principle.
\end{proof}

\begin{proposition}\label{sub} Any solution of (\ref{system}), (\ref{boundary})
 satisfies $q_0\le w_0, q_1\le w_1$.
\end{proposition}

\begin{proof}  Let us begin with $q_0$.  From  (\ref{system}) and Lemma \ref{h2} we have
\[
(w_0)_{z\zbar} = e^{2w_0} - e^{(w_0+w_1)-2w_0} \le 
e^{2w_0} - e^{h-2w_0}.
\]
With reference to Lemma \ref{q}, let $q_{0,\eps}$ be the solution of 
$(q_{0,\eps})_{z\zbar}= e^{2q_{0,\eps}} - e^{h-2q_{0,\eps}}$ subject to the boundary conditions 
$q_{0,\eps}(z)=(\ga_0+\eps+o(1))\log\vert z\vert$ 
as $\vert z\vert\to 0$, 
$q_{0,\eps}(z) \to 0$ as $\vert z\vert\to \infty$.
We claim that $w_0\ge q_{0,\eps}$, from which the desired result $w_0\ge q_0$ will follow by letting $\eps\downarrow 0$.

We have
\begin{align*}
(w_0-q_{0,\eps})_{z\zbar} &\le 
e^{2w_0} - e^{h-2w_0} - e^{2q_{0,\eps}}
+ e^{h-2q_{0,\eps}}
\\
&= e^{2w_0} - e^{2q_{0,\eps}} + e^h(e^{-2q_{0,\eps}} - e^{-2w_0}).
\end{align*}
By the maximum principle, we deduce that $w_0\ge q_{0,\eps}$, hence also $w_0\ge q_0$.
A similar argument shows that $w_1\ge q_1$.
\end{proof}

Now we are in a position to prove Theorem \ref{mainresult}.

\begin{proof}[Proof of Theorem \ref{mainresult}] Let $(\ga_0,\ga_1)\in\R^2$ satisfy the conditions $0\le\ga_0<2+\ga_1$ and $0\le\ga_1<1$.

\no{\em Step 1: Iteration scheme.} 

We shall construct  $(w_0^{(n)},w_1^{(n)})$ for $n=0,1,\dots$ converging to the desired solution.  For small values of $(\ga_0,\ga_1)$ we can obtain this solution if we start with 
$(w_0^{(0)},w_1^{(0)})=(0,0)$, but in general it will be necessary to start with $(w_0^{(0)},w_1^{(0)})=(g_0,g_1)$
for some previously constructed solution $(g_0,g_1)$.  Therefore, to set up the iteration scheme, we begin by assuming that we have a solution $(g_0,g_1)$ of
\begin{equation*}
\begin{cases}
(g_0)_{z\zbar}&= \ e^{2g_0} - e^{g_1-g_0} 
\\
(g_1)_{z\zbar}&= \ e^{g_1-g_0} - e^{-2g_1}
\end{cases}
\end{equation*}
with
\begin{equation*}
\begin{cases}
\ g_i(z) =\tilde\ga_i \log \vert z\vert+O(1)\
\text{as}\ \vert z\vert\to 0
\\
\ g_i(z) \to 0\ 
\text{as}\ \vert z\vert\to \infty
\end{cases}
\end{equation*}
such that 
\begin{equation}\label{g1}
w_0\le g_0, w_1\le g_1\ \text{for any solution $(w_0,w_1)$ of (\ref{system}) and (\ref{boundary}).}
\end{equation}
Furthermore, we shall assume that 
\begin{equation}\label{g2}
0\le \tilde\ga_0<\ga_0, 0\le \tilde\ga_1<\ga_1
\ \text{and also}\  \tilde\ga_1>\ga_0-2.
\end{equation}
For example, when $\ga_0\le 2$, (\ref{g1}) and (\ref{g2}) are satisfied by $(\tilde\ga_0,\tilde\ga_1)=(0,0)$.

Set $(w_0^{(0)},w_1^{(0)})=(g_0,g_1)$.
For $n\ge 0$ we define 
$(w_0^{(n+1)},w_1^{(n+1)})$ inductively as follows:
\begin{equation}\label{w0}
\begin{cases}
(w_0^{(n+1)})_{z\zbar} - (2+e^{g_1-q_0}) w_0^{(n+1)} =
f_0(w_0^{(n)},w_1^{(n)},z)
\\
\ w_0^{(n+1)}(z)=\ga_0\log\vert z\vert+O(1)
\text{ at $0$}, 
\ w_0^{(n+1)}(z) \to 0 \text{ at $\infty$}
\end{cases}
\end{equation}
where $f_0(u_0,u_1,z) =
e^{2u_0} - e^{u_1-u_0} - (2+e^{g_1-q_0})u_0$;
\begin{equation}\label{w1}
\begin{cases}
(w_1^{(n+1)})_{z\zbar} - (e^{g_1-q_0}+2e^{-2q_1}) w_1^{(n+1)} =
f_1(w_0^{(n)},w_1^{(n)},z)
\\
\ w_1^{(n+1)}(z)=\ga_1\log\vert z\vert+O(1)
\text{ at $0$}, 
\ w_1^{(n+1)}(z) \to 0
\text{ at $\infty$}
\end{cases}
\end{equation}
where $f_1(u_0,u_1,z) =
e^{u_1-u_0} - e^{-2u_1} - (e^{g_1-q_0}+2e^{-2q_1})u_1$.

When $z$ is small we have $e^{g_1-q_0}=O(\vert z\vert^{-\al})$,  $e^{-2q_1}=O(\vert z\vert^{-\be})$ for some $\al,\be\in (0,2)$, because $\ga_0-\tilde\ga_1<2$ and $\ga_1<1$.  The existence and uniqueness of $w_0^{(n+1)},w_1^{(n+1)}$ now follows from standard linear elliptic p.d.e.\  theory.  
The exponential decay of $w_i^{(n+1)}$ at infinity follows from that of $w_i^{(n)}$.

We must show that
\begin{equation}\label{a}
w_0^{(n+1)} \le w_0^{(n)} \text{ and } 
w_1^{(n+1)} \le w_1^{(n)} 
\end{equation}
and
\begin{equation}\label{b}
q_0\le w_0^{(n+1)}\text{ and } 
q_1\le w_1^{(n+1)}
\end{equation}
for $n\ge 0$.

\no{\em The case $n=0$. } 

From (\ref{w0}), $w_0^{(1)}$ is the solution of
\begin{equation*}
\begin{cases}
(w_0^{(1)})_{z\zbar} - (2+e^{g_1-q_0}) w_0^{(1)} =
F_0
\\
\ w_0^{(1)}(z)=\ga_0\log\vert z\vert+O(1)
\text{ at $0$}, 
\ w_0^{(1)}(z) \to 0 
\text{ at $\infty$}
\end{cases}
\end{equation*}
where $F_0(z) = f_0(g_0,g_1,z)=
e^{2g_0} - e^{g_1-g_0} - (2+e^{g_1-q_0})g_0$. Note that $g_0$ satisfies the same equation, but with different boundary conditions:
\begin{equation*}
\begin{cases}
(g_0)_{z\zbar} - (2+e^{g_1-q_0}) g_0 =
F_0
\\
\ g_0(z)=\tilde\ga_0\log\vert z\vert+O(1)
\text{ at $0$}, 
\ g_0(z) \to 0 
\text{ at $\infty$}.
\end{cases}
\end{equation*}
By the maximum principle, we deduce that
$
w_0^{(1)} \le g_0.
$

Similarly, from
\begin{equation*}
\begin{cases}
(w_1^{(1)})_{z\zbar} - (e^{g_1-q_0}+2e^{-2q_1}) w_1^{(1)} =
F_1
\\
\ w_1^{(1)}(z)=\ga_1\log\vert z\vert+O(1)
\text{ at $0$}, 
\ w_1^{(1)}(z) \to 0 
\text{ at $\infty$}
\end{cases}
\end{equation*}
where $F_1(z) =f_1(g_0,g_1,z)=
e^{g_1-g_0} - e^{-2g_1} - (e^{g_1-q_0}+2e^{-2q_1})g_1$,
and
\begin{equation*}
\begin{cases}
(g_1)_{z\zbar} - (e^{g_1-q_0}+2e^{-2q_1}) g_1 =
F_1
\\
\ g_1(z)=\tilde\ga_1\log\vert z\vert+O(1)
\text{ at $0$}, 
\ g_1(z) \to 0 
\text{ at $\infty$}
\end{cases}
\end{equation*}
we deduce that 
$
w_1^{(1)} \le g_1.
$

To prove (\ref{b}) for $n=0$, we note that 
\begin{align*}
(q_0)_{z\zbar} - (2+e^{g_1-q_0})q_0 
&\ge e^{2q_0} - e^{q_1-q_0} - (2+e^{g_1-q_0})q_0 
\ \text{as $h\le q_0+q_1$}
\\
&\ge e^{2q_0} - e^{g_1-q_0} - (2+e^{g_1-q_0})q_0 
\ \text{as $q_1\le g_1$}
\\
&\ge e^{2g_0} - e^{g_1-g_0} - (2+e^{g_1-q_0})g_0
=F_0.
\end{align*}
The last inequality follows from the fact that
\[
\tfrac{\b}{\b t}\left(
e^{2t} - e^{g_1-t} - (2+e^{g_1-q_0})t\right)
=
2(e^{2t}-1) + (e^{g_1-t} - e^{g_1-q_0})\le0
\]
whenever $q_0\le t\le 0$;  since
$q_0\le g_0\le 0$, we can put $t=g_0$.
Thus, $q_0$ satisfies the differential inequality
\begin{equation}\label{gq0}
(q_0)_{z\zbar} - (2+e^{g_1-q_0})q_0 \ge F_0.
\end{equation}
By the maximum principle, we deduce that
$
q_0\le w_0^{(1)}.
$

Similarly, we can obtain
\begin{equation}\label{gq1}
(q_1)_{z\zbar} - (e^{g_1-q_0}+2e^{-2q_1})q_1 \ge F_1.
\end{equation}
by using the fact that
\[
\tfrac{\b}{\b t}\left(
e^{t-q_0} - e^{-2t} - (e^{g_1-q_0}+2e^{-2q_1})t\right)
=
e^{t-q_0} - e^{g_1-q_0} + 2e^{-2t} - 2e^{-2q_1} \le0
\]
whenever $q_1\le t\le g_1$.  Applying the maximum principle again, we have
$
q_1\le w_1^{(1)}.
$

This completes the proof of (\ref{a}) and (\ref{b}) for $n=0$.

\no{\em The  inductive step from $n$ to $n+1$.} 

From the definitions of $f_0,f_1$ in (\ref{w0}), (\ref{w1}),
we see that $\tfrac{\b f_0}{\b u_1}(u_0,u_1,z)<0$, and that
\[
\tfrac{\b f_0}{\b u_0}(u_0,u_1,z)=
2(e^{2u_0}-1) + e^{u_1-u_0} - e^{g_1-q_0}\le0
\]
whenever $q_0\le u_0\le 0$ and $u_1\le g_1$. Thus 
\[
f_0(w_0^{(n-1)},w_1^{(n-1)},z)\le 
f_0(w_0^{(n-1)},w_1^{(n)},z)\le
f_0(w_0^{(n)},w_1^{(n)},z),
\]
as $q_i\le w_i^{(n)}\le w_i^{(n-1)}\le g_i$ by the inductive hypothesis.  The maximum principle then gives 
$w_0^{(n+1)}\le w_0^{(n)}$.  

Similarly, $\tfrac{\b f_1}{\b u_0}(u_0,u_1,z)<0$, and
\[
\tfrac{\b f_1}{\b u_1}(u_0,u_1,z)=
e^{u_1-u_0}  - e^{g_1-q_0} + 2e^{-2u_1} - 2e^{-q_1}\le 0
\]
whenever $q_1\le u_1\le g_1$ and $q_0\le u_0$. As we are assuming $q_i\le w_i^{(n)}\le w_i^{(n-1)}\le g_i$, we obtain
\[
f_1(w_0^{(n-1)},w_1^{(n-1)},z)\le 
f_1(w_0^{(n)},w_1^{(n-1)},z)\le
f_1(w_0^{(n)},w_1^{(n)},z).
\]
The maximum principle gives 
$w_1^{(n+1)}\le w_1^{(n)}$. 
This completes the inductive step for (\ref{a}). 

To prove (\ref{b}) for $w_0^{(n+1)}$, we note that 
\begin{align*}
f_0(w_0^{(n)},w_1^{(n)},z) \le \cdots &\le
f_0(w_0^{(0)},w_1^{(0)},z) = F_0 \ \text{as above}
\\
&\le
(q_0)_{z\zbar} - (2 + e^{g_1-q_0})q_0 \ \text{by (\ref{gq0})}
\end{align*}
and similarly (\ref{gq1}) implies
\begin{align*}
f_1(w_0^{(n)},w_1^{(n)},z) \le \cdots 
&\le f_1(w_0^{(0)},w_1^{(0)},z) = F_1
\\
&\le
(q_1)_{z\zbar} - (e^{g_1-q_0}+2e^{-q_1})q_1.
\end{align*}
By the maximum principle, it follows that
$q_0\le w_0^{(n+1)}$ and $q_1\le w_1^{(n+1)}$
as required.
This completes the proof of (\ref{a}) and (\ref{b}).

Elliptic estimates show that the sequence  $w_i^{(n)}$ converges to some
$w_i\in$ \mbox{$C^\infty(\R^2\setminus\{(0,0)\})$}.  Clearly these $w_0,w_1$ satisfy (\ref{system}), (\ref{boundary}). 

\no{\em Step 2: Existence of maximal solution when $0\le\ga_0<2$ and $0\le\ga_1<1$.}

Let us take $(g_0,g_1)=(0,0)$ and $(\tilde\ga_0,\tilde\ga_1)=(0,0)$  in Step 1. Conditions (\ref{g1}) and (\ref{g2})
are satisfied, so we obtain a solution $(w_0,w_1)$ from the iteration.

We claim that $(w_0,w_1)$ is in fact a maximal solution, i.e.\ 
$v_i\le w_i$ for any other solution $(v_0,v_1)$.  To prove this, we shall show by induction that
$v_i\le w_i^{(n)}$ for all $n$, then take the limit $n\to\infty$. 
By Propositions \ref{super} and \ref{sub}, we have
$q_i\le v_i\le w_i^{(0)}$.  By the inductive hypothesis $v_i\le w_i^{(n)}$, and the fact that $f_i$ is decreasing (see Step 1), we have
$f_i(v_0,v_1,z)  
\ge f_i(w_0^{(n)},w_1^{(n)},z)$. Then the maximum principle gives $v_i\le w_i^{(n+1)}$, as required.  

\no{\em Step 3: Uniqueness when $0\le\ga_0<2$ and $0\le\ga_1<1$.} 

For any $r>0$ let
$B_r=\{ x\in\R^2\st 0\le \vert x\vert \le r\}$, and
for $R>r>0$ let
$B_{r,R}=\{ x\in\R^2\st r\le \vert x\vert \le R\}$. The boundary of $B_{r,R}$ will be written as
$\b B_{r,R} = \b B_R - \b B_{r}$ below.

We multiply the system (\ref{system}) by
$x\cdot\nabla w_0$ and integrate over 
$B_{r,R}$. 
For $w_0$, the left hand side gives
\begin{align*}
&\smallint_{\!B_{r,R}}\ (x\cdot\nabla w_0) \De w_0 \,dx
\\
&=-\smallint_{\!B_{r,R}}\ \vert \nabla w_0\vert^2 \,dx
-\tfrac12 \smallint_{\!B_{r,R}}\ x\cdot\nabla \vert \nabla w_0\vert^2 \,dx
+ \smallint_{\b B_R-\b B_r}\ (x\cdot\nabla w_0)\tfrac{\b w_0}{\b \nu} \,ds
\\
&= -\tfrac12 \smallint_{\b B_R-\b B_r}\ x\cdot\nu \vert \nabla w_0\vert^2 \,ds
+ \smallint_{\b B_R-\b B_r}\ (x\cdot\nabla w_0)\tfrac{\b w_0}{\b \nu} \,ds
\end{align*}

\no When $R\to\infty$ and $r\to 0$, we have
\[
\left\vert 
\smallint_{\!\b B_R}\  x\cdot\nu \vert \nabla w_0\vert^2 \,ds
\right\vert,
\
\left\vert 
\smallint_{\!\b B_R}\ (x\cdot\nabla w_0)\tfrac{\b w_0}{\b \nu} \,ds
\right\vert
\ 
\to 0
\]
and
\[
\tfrac12 \smallint_{\!\b B_r} \ x\cdot\nu \vert \nabla w_0\vert^2 \,ds
-
\smallint_{\!\b B_r}\ (x\cdot\nabla w_0)\tfrac{\b w_0}{\b \nu} \,ds
\to -\pi \ga_0^2.
\]
Thus
\[
\smallint_{\!\R^2\,}\ (x\cdot\nabla w_0) \De w_0 \,dx=-\pi \ga_0^2.
\]
Similarly for $w_1$ we find that 
\[
\smallint_{\!\R^2\,}\ (x\cdot\nabla w_1) \De w_1 \,dx=-\pi \ga_1^2.
\]
Multiplying the right hand sides of  (\ref{system}) by $x\cdot\nabla w_0$ and
$x\cdot\nabla w_1$ (respectively), subtracting them, and integrating, we obtain
\begin{align*}
&-\smallint_{\!\R^2\,}\ (x\cdot\nabla w_0)(e^{2w_0}-e^{w_1-w_0})\,dx
-\smallint_{\!\R^2\,} \ (x\cdot\nabla w_1)(e^{w_1-w_0}-e^{-2w_1})\,dx
\\
&=
\tfrac12 
\smallint_{\!\R^2\,}\ x\cdot\nabla (1-e^{2w_0})
+
 x\cdot\nabla (1-e^{w_1-w_0})
+
\tfrac12 
 x\cdot\nabla (1-e^{-2w_1})\,dx
\\
&=
\smallint_{\!\R^2\,}\ -(1-e^{2w_0})-2(1-e^{w_1-w_0})
-(1-e^{-2w_1})
\,dx.
\end{align*}
We deduce that
\[
\smallint_{\!\R^2\,}\ (1-e^{2w_0})+2(1-e^{w_1-w_0})
+(1-e^{-2w_1})
\,dx = \pi(\ga_0^2+\ga_1^2).
\]
On the other hand, integrating the right hand sides directly, we obtain
\begin{gather*}
\smallint_{\!\R^2\,}\ -(1-e^{2w_0}) + (1-e^{w_1-w_0})
\,dx = -2\pi \ga_0
\\
\smallint_{\!\R^2\,}\ -(1-e^{w_1-w_0}) + (1-e^{-2w_1})
\,dx = -2\pi \ga_1.
\end{gather*}

Thus, we obtain the identities
\begin{equation}\label{p0}
\smallint_{\!\R^2\,}\ (1-e^{2w_0}) \,dx
=\tfrac\pi4(\ga_0^2+\ga_1^2+6\ga_0+2\ga_1)
\end{equation}
and
\begin{equation}\label{p1}
\smallint_{\!\R^2\,}\ (1-e^{-2w_1}) \,dx
=\tfrac\pi4(\ga_0^2+\ga_1^2-6\ga_1-2\ga_0).
\end{equation}
These imply uniqueness of the solution $(w_0,w_1)$.  Namely, if $(v_0,v_1)$ is another
solution, then  (\ref{p0}) and (\ref{p1}) show that
$\smallint_{\!\R^2\,}\ (1-e^{2w_0}) \,dx = 
\smallint_{\!\R^2\,}\ (1-e^{2v_0}) \,dx$ and
$\smallint_{\!\R^2\,}\ (1-e^{-2w_1}) \,dx =
\smallint_{\R^2\,} (1-e^{-2v_1}) \,dx$. But $w_i$ is maximal, so it must coincide with $v_i$.

\no{\em Step 4: The case $0\le\ga_0<2+\ga_1$ and $0\le\ga_1<1$.}

We may assume that $2\le\ga_0<2+\ga_1$ and $0<\ga_1<1$, otherwise we are in the situation of Step 2.  Let us choose any $(\tilde\ga_0,\tilde\ga_1)$ such that
\begin{equation}\label{step2}
0\le\tilde\ga_0<2,\ 0\le\tilde\ga_1<1
\end{equation}
and
\begin{equation}\label{push}
0\le \tilde\ga_0<\ga_0,\ 
\ga_0-2 < \tilde\ga_1 < \ga_1.
\end{equation}
By (\ref{step2}), we have a solution $(\tilde w_0,\tilde w_1)$ from Step 2.  
In Proposition \ref{comparable} below, we shall prove that
\[
w_i\le \tilde w_i
\]
for any solution $(w_0,w_1)$ of (\ref{system}), (\ref{boundary}) with $\tilde\ga_0\le\ga_0<2+\ga_1$ and $\tilde\ga_1<\ga_1<1$. Hence we may take $(g_0,g_1)=
(\tilde w_0,\tilde w_1)$ as the starting point for the iteration in Step 2, and obtain a solution $(w_0,w_1)$ of (\ref{system}), (\ref{boundary}). The method of Step 2  shows that the solution is maximal.  By applying the Pohozaev identity of Step 3, we see that uniqueness holds in this case also.

\no{\em Step 5: The case $0\le\ga_0\le2+\ga_1$ and $0\le\ga_1\le1$.}

So far we have treated the case where $0\le\ga_0<2+\ga_1$ and $0\le\ga_1<1$. Next we consider the boundary case where equality may hold on the right hand sides of these inequalities.

Without loss of generality, we may assume that there exist
sequences $\ga^{(n)}_0 \uparrow \ga_0$, 
$\ga^{(n)}_1 \uparrow \ga_1$, such that each $(\ga^{(n)}_0,\ga^{(n)}_1)$ is in the range for Step 4.  Let 
$(w_0^{(n)},w_1^{(n)})$ be the solution corresponding to
$(\ga^{(n)}_0,\ga^{(n)}_1)$.  Then 
$w_i^{(n+1)}\le w_i^{(n)}$ for $i=0,1$ by 
Proposition \ref{comparable}. Since $q_i\le w^{(n)}_i$, the sequence 
$w_i^{(n)}$ converges to some $w_i$ in $C^\infty(\C\setminus\{0\})$.  

By Proposition \ref{comparable}, $w_i$ is bounded above by
$w_i^{(n)}$. Thus $w_i$ is
a maximal solution, so the Pohozaev identity argument of Step 3 can be used again to show that $w_i$ is the unique solution satisfying the boundary conditions
$w_i(z) = (\ga_i+o(1)) \log \vert z\vert$ as $\vert z\vert\to 0$ and $w_i(z) \to 0$ as $\vert z\vert\to \infty$.

This completes the proof of Theorem \ref{mainresult}.
\end{proof}

Finally, we give the following result which was used in Step 4.

\begin{proposition}\label{comparable}
Let $(w_0,w_1)$, $(\tilde w_0,\tilde w_1)$ be solutions of (\ref{system}) with boundary conditions corresponding (respectively) to 
$(\ga_0,\ga_1)$, $(\tilde\ga_0,\tilde\ga_1)$. If
(\ref{step2}) and (\ref{push}) are satisfied, then 
$w_i\le \tilde w_i$
$(i=0,1)$.
\end{proposition}

\begin{proof} Let $(\tilde w_0^{(n)},\tilde w_1^{(n)})$
denote the solution of the monotone scheme in Step 1 with $(g_0,g_1)=(0,0)$ and the boundary conditions 
$\tilde w_i^{(n)}=\tilde\ga_i\log\vert z\vert+O(1)$ as $\vert z\vert\to 0$ and $\tilde w_i^{(n)}\to 0$ as 
$\vert z\vert\to \infty$. Let $(\tilde q_0,\tilde q_1)$ be as in Lemma \ref{q} with boundary conditions given by 
$(\tilde \ga_0,\tilde \ga_1)$. By Step 2 we have
$\lim_{n\to\infty} \tilde w_i^{(n)} = w_i$.

We have $w_0,w_1\le 0$, i.e.\ 
$w_0\le \tilde w_0^{(0)}$, $w_1\le \tilde w_1^{(0)}$.
We claim that, 
\[
\text{ if }
w_0\le \tilde w_0^{(n)}\text{ and }
w_1\le \tilde w_1^{(n)},
\text{ then }
w_0\le \tilde w_0^{(n+1)}\text{ and }
w_1\le \tilde w_1^{(n+1)}.
\]
If $w_0>\tilde w_0^{(n+1)}$ at some point, then
\[
w_0(z_0) -\tilde w_0^{(n+1)}(z_0)
= \max\,(w_0 - \tilde w_0^{(n+1)})> 0
\]
for some  in  $z_0\in \C\setminus\{0\}$.
Thus, $\tilde q_0(z_0) < w_0(z_0) \le \tilde w_0^{(n)}(z_0)
\le 0$ and $w_1(z_0)\le  \tilde w_1^{(n)}(z_0)
\le 0$. This implies 
$f_0(w_0(z_0),w_1(z_0),z_0)\ge 
f_0(\tilde w_0^{(n)}(z_0),\tilde w_1^{(n)}(z_0),z_0)$.
Noting that 
\[
(w_0)_{z\zbar} - (2 + e^{g_1-q_0})w_0
= 
f_0(w_0,w_1,z_0)\ge 
f_0(\tilde w_0^{(n)},\tilde w_1^{(n)},z_0),
\]
we obtain a contradiction by using the maximum principle. 

From the fact that
$f_1(w_0,w_1,z)\ge 
f_1(\tilde w_0^{(n)},w_1,z)\ge
f_1(\tilde w_0^{(n)},\tilde w_1^{(n)},z)$
whenever 
$w_1\ge \tilde w_1^{(n+1)} \ge \tilde q_1$, a similar
argument leads to a contradiction if 
$w_1>\tilde w_1^{(n+1)}$ at some point of
$\C\setminus\{0\}$.
This establishes the claim.  The proposition follows immediately from this.
\end{proof}

\begin{remark}\label{concluding}
{\bf (i)\,}
As pointed out in Remark \ref{derivatives}, our system 
(\ref{system}) has the properties
\begin{gather*}
\tfrac{\b}{\b w_1} \left(
e^{2w_0} - e^{w_1-w_0}
\right)< 0
\\
\ \tfrac{\b}{\b w_0} \left(
e^{w_1-w_0} - e^{-2w_1}
\right)< 0.
\end{gather*}
These imply the properties $\tfrac{\b f_0}{\b w_1}<0$,
$\tfrac{\b f_1}{\b w_0}<0$ which were used in Step 1 of the proof of Theorem \ref{mainresult}. On the other hand, the properties
\begin{gather*}
\tfrac{\b}{\b w_0} \left(
e^{2w_0} - e^{w_1-w_0}
\right)
\to\infty
\ \text{as}\  w_0\to-\infty
\\
\ \tfrac{\b}{\b w_1} \left(
e^{w_1-w_0} - e^{-2w_1}
\right)
\to\infty
\ \text{as}\  w_1\to-\infty
\end{gather*}
cause difficulties in the monotone scheme. 
To remedy this, we subtracted linear terms from both sides of the system in order to have 
$\tfrac{\b f_0}{\b w_0}<0$,
$\tfrac{\b f_1}{\b w_1}<0$.
However, the proof of the existence of maximal solutions throughout the full range of $\ga_0,\ga_1$ is more technical than in the case where no singularity exists at $0$.

{\bf (ii)\,}
The method used to prove the existence statement of Theorem \ref{mainresult} can be extended
to systems of the form
\begin{equation*}
\begin{cases}
\ (w_0)_{z\bar z}&= \ e^{2w_0} - e^{w_1-w_0} 
+2\pi \sum_{j=1}^{N_0} \ga_0^{(j)} \de_{p_j}
\\
\ (w_1)_{z\bar z}&= \ e^{w_1-w_0} - e^{-2w_1}
+2\pi \sum_{j=1}^{N_1} \ga_1^{(j)} \de_{q_j}
\end{cases}
\end{equation*}
where $\de_p$ denotes the Dirac measure at $p$.
Theorem \ref{mainresult} is the case $N_0=N_1=1$ and
$p_1=q_1=(0,0)$.  
\end{remark}

\section{Relation with the field-theoretic solutions}\label{tt}

We shall show in this section that our distinguished two-parameter family of solutions of the \stoda includes a finite number of even more distinguished solutions, corresponding to quantum cohomology or Landau-Ginzburg models. As a result, these models can be said to possess \ll global\rr $tt^\ast$ structures. These models can be specified by certain holomorphic matrix-valued functions, which we interpret as 
holomorphic data for solutions of the \sstoda.  

In section \ref{toda} we described the \sstoda, but we have not yet described the relation between solutions and holomorphic data. In the appendix we review this well known relation in the case of the usual Toda lattice. Here we shall just explain the modifications needed for the \sstoda.

\no{\em (i) Holomorphic data for solutions of the \sstoda.}

The local correspondence between solutions of the \stoda and their holomorphic data works in exactly the same way as for the usual Toda lattice in Theorems \ref{holtotoda} and \ref{todatohol}.   As in the appendix, we have a chain of correspondences
\[
p_0,\dots,p_n 
\ \leftrightarrow\ 
\eta
\ \leftrightarrow\ 
L
\ \leftrightarrow\ 
F, B
\ \leftrightarrow\ 
b_0,\dots,b_n
\ \leftrightarrow\ 
w_0,\dots,w_n
\]
but new features are the choice of holomorphic functions $h_0,\dots,h_n$ which relate $b_i$ and $w_i$ in formula (\ref{wi}), and the Lie groups involved in the Iwasawa factorization. 

As in the
appendix, we must choose $h_0,\dots,h_n$ such that all $\nu_i$ are equal, say $\nu_i=\nu$ for all $i$, which implies that $\nu^{n+1}=p_0\dots p_n$ and $\nu=p_i h_i/h_{i-1}$.  However, for the \stoda we have the condition $h_ih_j=1$ whenever $w_i+w_j=0$ (see the proof of Corollary \ref{twofunctions} and Table 1). This determines $h_0,\dots,h_n$ in terms of $p_0,\dots,p_n$.   

It is the global --- not just local --- aspects of this correspondence that interest us,  and here
there is a significant new phenomenon:  if the real form of the Lie group is not compact, then the Iwasawa factorization $L=FB$ is not guaranteed to exist on the entire domain of $L$.   For the standard Toda lattice,  it is known that
\[
\La \SLPC = \La \SU_{n+1} \ \La_+ \SLPC
\]
(the same holds when the twisting conditions
$\tau(f(\la))=f(\oom \la)$ and $\si(f(\la))=f(-\la)$ are imposed on both sides). 
For the noncompact group
$\SL_{n+1}^N\R$, however, it is known only that 
$\La \SLPC$ contains
\[
\La \SL^N_{n+1}\R \ \La_+ \SLPC
\]
as an open subspace. 
Since $I$ is contained in this subspace, if we assume a basepoint condition of the form $L(z_0)=I$, then the  Iwasawa factorization exists on some neighbourhood of $z_0$, but in general it is very difficult to predict how large this neighbourhood can be. 
In the case of the usual Toda lattice, if $L$ is holomorphic on $\C$, then $w_0,\dots,w_n$ are smooth on $\C$, but we cannot make this inference in the case of the \sstoda. 
This is where we shall need Theorem \ref{mainresult}.

Another difficulty we face with the field-theoretic solutions is that the natural basepoint is $z_0=0$, which may be a singular point of the holomorphic data $p_0,\dots,p_n$. Additional arguments (cf.\ \cite{DoGuRoXX}, \cite{IrXX}) are needed to deal with this.

\no{\em (ii) Holomorphic data for radially-invariant solutions.}

In view of the following observation, we shall restrict attention to radially-invariant solutions of the \sstoda.

\begin{proposition} 
The solutions $w_0,w_1$ given in Theorem \ref{mainresult}
are radially-invariant,
i.e.\ $w_i(z,\zbar)=w_i(\vert z\vert)$ for $i=0,1$.
\end{proposition} 

\begin{proof} If some solution were not radially-invariant, rotation of the parameter $z$ would produce new solutions satisfying the same asymptotic conditions. This would contradict the uniqueness statement of 
Theorem \ref{mainresult}.
\end{proof}

It turns out that the holomorphic data $\eta$ for such solutions has the special form $p_i=c_i z^{k_i}$ for some constants $c_i, k_i$. To see this, we shall make use of the \ll homogeneity\rr property
\begin{equation}\label{star}
\tfrac1{\eps\la} \eta(\eps^a z) d(\eps^a z)
=
 T(\eps)^{-1}\   \tfrac1\la\eta(z) dz\  T(\eps)
\ \ \text{ for all $\eps\in S^1$}
\end{equation}
where $T(\eps)=\diag(1,\eps^{e_1},\dots,\eps^{e_n})$
for some constants $e_1,\dots,e_n$. Under mild conditions, this characterizes the special potentials:

\begin{proposition} 
If condition (\ref{star}) holds for some $a,e_1,\dots,e_n$ such that $a\ne 0$, then $p_i=c_i z^{k_i}$ for all $i$, where the $c_i$ are constants and $k_i=(e_{i-1}-e_i + 1-a)/a$.  Conversely, if  $p_i=c_i z^{k_i}$ for some
$c_0,\dots,c_n$ and for some $k_0,\dots,k_n$ such that $ n+1+\sum_{i=0}^n k_i \ne 0$, then condition (\ref{star}) holds 
with $e_i=-a(k_1+\cdots +k_i)+i(1-a)$, $a=(n+1)/( n+1+\sum_{i=0}^n k_i)$.
\end{proposition} 

\begin{proof} Condition  (\ref{star}) is equivalent to
$\eps^{a-1} p_i(\eps^a z) = \eps^{e_{i-1}-e_i} p_i(z)$ for all $i$, which gives relations between
$k_0,\dots,k_n$ and $e_1,\dots,e_n$.
Both assertions follow directly from this.
\end{proof}

For our purposes in part (iii) below, a restricted set of holomorphic data will suffice. The theorem generalizes special cases which have appeared in \cite{BoIt95} and 
\cite{DoGuRoXX}.

\begin{theorem}\label{correspondence} 
Let $c_0,\dots,c_n$ and $k_0,\dots,k_n$ be real numbers such that $c_i>0$,  $k_i\ge -1$ for all $i$, and 
$\sum_{i=0}^n k_i > -(n+1)$.  
Then the   holomorphic data
\[
p_0=c_0 z^{k_0},\dots,p_n=c_n z^{k_n}
\]
gives a
radially-invariant solution  $w_0(t),\dots,w_n(t)$
of the \stoda in a punctured neighbourhood of $t=0\in\C$.  
We have\footnote{In this theorem and its proof, the formula
$w_i(t) = (\ga_i+o(1)) \log \vert t\vert$
refers to $w_i$ of section \ref{solutions}. 
}
\[
w_i(t) = (\ga_i+o(1)) \log \vert t\vert
\quad\text{as  $\vert t\vert\to 0$,}
\]
where $\ga_i$ is a certain rational function of $k_0,\dots,k_n$ (independent of $c_0,\dots,c_n$).  For the system (\ref{generalsystem}) of Corollary \ref{twofunctions} where $w,v$ are as in Table 1, 
these rational functions are listed in Table 2. 
\end{theorem}

\begin{table}[ht]
\renewcommand{\arraystretch}{1.5}
\begin{tabular}{cc|c|c}
$l_1$ & $l_2$ & $\gamma_0$ & $\gamma_1$
 \\
\hline
$4$ &  &  $\frac{3k_0 - 2k_1 - k_2}{k_0 + 2k_1 + k_2+4}$&
$\frac{k_0 + 2k_1 - 3k_2}{k_0 + 2k_1 + k_2+4}$
\\
$5$ &  & $\frac{4k_0 - 2k_1 -2k_2}{k_0 + 2k_1 + 2k_2+5}$ &
${\scriptstyle2}
\frac{k_0 + 2k_1 -3k_2}{k_0 + 2k_1 + 2k_2+5}$
\\
$1$ & $4$&  
${\scriptstyle2}
\frac{3k_0 - 2k_2 - k_3}{2k_0 + 2k_2 + k_3+5}$
&$\frac{2k_0 + 2k_2 - 4k_3}{2k_0 + 2k_2 + k_3+5}$
\\
$1$ & $5$ & ${\scriptstyle2}
\frac{4k_0 -2k_2 -2k_3}{2k_0 + 2k_2 + 2k_3+6}$ &${\scriptstyle2}
\frac{2k_0 +2k_2 -4k_3}{2k_0 + 2k_2 + 2k_3+6}$
\\
$2$ & $2$ &  $\frac{-2k_0 - k_1 + 3k_3}{2k_0 + k_1 + k_3+4}$&
$\frac{2k_0 - 3k_1 + k_3}{2k_0 + k_1 + k_3+4}$
\\
$2$ & $3$ &  ${\scriptstyle2}
\frac{-2k_0 -k_1 +3k_3}{2k_0 + k_1 + 2k_3+5}$&$\frac{2k_0 -4k_1 +2k_3}{2k_0 + k_1 + 2k_3+5}$
\\
$3$ & $2$ &  $\frac{-2k_0 - 2k_1 +4k_4}{2k_0 + 2k_1 + k_4+5}$&
${\scriptstyle2}
\frac{2k_0 - 3k_1 +k_4}{2k_0 + 2k_1 + k_4+5}$
\\
$3$ & $3$ & ${\scriptstyle2}
\frac{-2k_0 -2k_1 +4k_4}{2k_0 + 2k_1 + 2k_4+6}$ &
${\scriptstyle2}
\frac{2k_0 -4k_1 +2k_4}{2k_0 + 2k_1 + 2k_4+6}$
\\
$4$ & $1$ &  ${\scriptstyle2}
\frac{3k_0 -2k_1 -k_2}{2k_0 + 2k_1 + k_2+5}$
&$\frac{2k_0 +2k_1 -4k_2}{2k_0 + 2k_1 + k_2+5}$
\\
$5$ & $1$ &  ${\scriptstyle2}
\frac{2k_0 -k_1 -k_2}{2k_0 + 2k_1 + 2k_2+6}$&${\scriptstyle2}
\frac{2k_0 +2k_1 -4k_2}{2k_0 + 2k_1 + 2k_2+6}$
\end{tabular}
\bigskip
\caption{}
\end{table}

\begin{proof} Let us assume first that $k_i\ge 0$ for all $i$.  
In this case the holomorphic data is defined at $z_0=0$, and we shall normalize $L$ by taking $L(0)=I$. Then the Iwasawa factorization $L=FB$ holds on a neighbourhood of $0$, and 
(by uniqueness of the Iwasawa factorization)
we have $F(0)=B(0)=I$. 
The homogeneity condition (\ref{star}) is inherited by $L$ (by the uniqueness property of local solutions of ordinary differential equations), and also by $F$ and $B$ (by uniqueness of the Iwasawa factorization).  It follows that the diagonal terms $b_0,\dots,b_n$  of $B_0$ satisfy
$b_i(\eps^a z)=b_i(z)$,
that is, they are radially-invariant.  By formula (\ref{wi}),  $w_0(t),\dots,w_n(t)$ are also radially-invariant. 
Since $b_0(t),\dots,b_n(t)$ are defined on a neighbourhood of $0$, formula (\ref{wi}) shows that
$w_0(t),\dots,w_n(t)$  of the \stoda are defined on a punctured neighbourhood of $0$.

The asymptotic expression for $w_i$ near $0$ may also be 
computed from formula (\ref{wi}).  We shall explain the computation in the case 
$(l_1,l_2)=(2,2)$ of Table 1; all other cases are similar.

From Table 1 we have  
$p_0=p_2$ (so it suffices to use $k_0,k_1,k_3$) and also
$h_0h_1=1$, $h_2h_3=1$.  
We wish to find the coefficients of $\log\vert t\vert$ in the asymptotic expressions for $w=2w_3$ and $v=2w_0$.  For this we need $h_3,h_0$ and also the 
change of variable formula $dt/dz = \nu = (p_0\dots p_n)^{\frac 1{n+1}} 
= (p_0^2p_1p_3)^{\frac14}$.  

We have $\nu=p_ih_i/h_{i-1}$ (see part (i) above), i.e.\ $h_i/h_{i-1}=\nu/p_i$.  In particular
\[
(h_0^{-2}=)\ 
\frac{h_1}{h_0}=\frac{\nu}{p_1} = 
p_0^{\frac 12} p_1^{-\frac 34} p_3^{\frac 14},\ \ 
(h_3^2=)\ 
\frac{h_3}{h_2}=\frac{\nu}{p_3} = 
p_0^{\frac 12} p_1^{\frac 14} p_3^{-\frac 34},
\]
hence
\[
h_3(z)=c_0^{\frac14}c_1^{\frac18}c_3^{-\frac38}
z^{(2k_0+k_1-3k_3)/8},\ \ 
h_0(z)=c_0^{-\frac14}c_1^{\frac38}c_3^{-\frac18}
z^{(-2k_0+3k_1-k_3)/8}.
\]
Next, from $dt/dz=(p_0^2p_1p_3)^{\frac14}=
c_0^{\frac12}c_1^{\frac14}c_3^{\frac14}
z^{(2k_0+k_1+k_3)/4}$, we obtain
\[
t=
\tfrac{4}{2k_0+k_1+k_3+4}\,
c_0^{\frac12}c_1^{\frac14}c_3^{\frac14}\,
z^{(2k_0+k_1+k_3+4)/4}
\]
hence
\[
z=
\left(
\tfrac{2k_0+k_1+k_3+4}{4}\,
c_0^{-\frac12}c_1^{-\frac14}c_3^{-\frac14}\,
t
\right)^{4/(2k_0+k_1+k_3+4)}.
\]
Substituting this into the above expressions for $h_3,h_0$, we obtain:
\begin{equation}\label{2w3}
2w_3 =
\tfrac{-2k_0-k_1+3k_3}{2k_0+k_1+k_3+4}\,\log\vert t\vert
+
K_3
-2\log\vert 
c_0^{\frac14}c_1^{\frac18} c_3^{-\frac38}
\vert
+O(t)
\quad
\end{equation}
\begin{equation}\label{2w0}
2w_0 =
\tfrac{2k_0-3k_1+k_3}{2k_0+k_1+k_3+4}\,\log\vert t\vert
+
K_0
-2\log\vert 
c_0^{-\frac14}c_1^{\frac38} c_3^{-\frac18}
\vert
+O(t)
\quad
\end{equation}
where
\[
K_3=
\tfrac{-2k_0-k_1+3k_3}{2k_0+k_1+k_3+4}\,
\log\left\vert
\tfrac{2k_0+k_1+k_3+4}{4\,c_0^{1/2}c_1^{1/4} c_3^{1/4}} 
\right\vert,
\  
K_0=
\tfrac{2k_0-3k_1+k_3}{2k_0+k_1+k_3+4}\,
\log\left\vert
\tfrac{2k_0+k_1+k_3+4}{4\,c_0^{1/2}c_1^{1/4} c_3^{1/4}} 
\right\vert.
\]
In particular, we obtain $w(t)=\ga_0\log\vert t\vert + O(1)$,
$v(t)=\ga_1\log\vert t\vert + O(1)$ 
(hence $w(t) = (\ga_0+o(1)) \log \vert t\vert$,
$v(t) = (\ga_1+o(1)) \log \vert t\vert$)
where
\[
\ga_0=\tfrac{-2k_0-k_1+3k_3}{2k_0+k_1+k_3+4},\ \ 
\ga_1=\tfrac{2k_0-3k_1+k_3}{2k_0+k_1+k_3+4}.
\]
This gives the fifth row of Table 2. The others can be obtained in a similar way.  

If $k_i=-1$ for at least one value of  $i$, then it is possible to find a solution of $L^{-1}dL=\tfrac1\la \eta dz$ such that 
$L$ admits an Iwasawa factorization
$L=FB$
in a punctured neighbourhood of $t=0\in\C$.  For the case $n=1$ the method of Theorem 4.1 of \cite{DoGuRoXX} applies. The general case may be proved in the same way, or by interpreting Theorem 3.7 of \cite{IrXX} 
in the language of loop groups.
The analogous calculation of the asymptotic behaviour of
$w_i$ (see Corollary 5.3 of \cite{DoGuRoXX}) gives $w_i(t) = (\ga_i+o(1)) \log \vert t\vert$ as 
$\vert t\vert\to 0$; the coefficients $\ga_i$ are given by exactly the same formulae as in the case $k_i\ge 0$.  
\end{proof}

This allows us to obtain holomorphic data for the solutions of the \stoda obtained in Theorem \ref{mainresult}, in the following way.  First, we choose a real number $k$.  
Then, we observe that $\ga_0,\ga_1$ determine unique $k_0,\dots,k_n$ such that $\sum_{i=0}^n k_i=k$.   For example, in the case $(l_1,l_2)=(2,2)$, we have
$-2k_0-k_1+3k_3 = \ga_0(k+4)$, 
$2k_0 - 3k_1 + k_3 = \ga_1(k+4)$, and
$2k_0+k_1+k_3 = k$
from which $\ga_0,\ga_1,k$ determine $k_0,k_1,k_3$. From the same equations, we see that if $k$ is sufficiently large then
$k_i\ge-1$ for all $i$.  Using 
$z^{k_0},\dots,z^{k_n}$ as \ll reference data\rr (in the sense of the discussion after Theorem \ref{holtotoda}), we obtain $\tilde F$, $\tilde L$, and $\tilde \eta$, all of which satisfy the homogeneity condition. Hence 
$\tilde p_0,\dots,\tilde p_n$ are necessarily of the form $c_0z^{k_0},\dots,c_nz^{k_n}$ for some $c_0,\dots,c_n$.

\no{\em (iii) Field-theoretic examples.}

The above results apply to several 
\ll field-theoretic examples\rrr.  
We list some quantum cohomology examples in Table 3, and some Landau-Ginzburg examples in Table 4. In both cases, the matrix $\eta$ giving the holomorphic data appears as the matrix of multiplication by a cyclic element of a certain algebra, namely the quantum cohomology algebra or the Milnor algebra (Jacobian algebra). We describe this construction very briefly.

The quantum cohomology of complex projective space $\C P^n$ and (orbifold) quantum cohomology of any weighted
complex projective space $P(w_0,\dots,w_n)$ provide 
holomorphic data of the type needed for local solutions of the \sstoda.  (Quantum cohomology of other manifolds or orbifolds give local solutions of the \ll $tt^\ast$-equations\rr but in general these will not be Toda-like in our sense.)  We just give a brief explanation here for the case $M=\C P^n$; the case $M=P(w_0,\dots,w_n)$ is very similar.
First, it is known that the (small) quantum cohomology algebra
$QH^\ast(\C P^n;\C)$ is isomorphic to 
$\C[x,q]/(x^{n+1}-q)$; where $x$ is a basis vector of $H^{2}(M;\C)\cong\C$ and $q$ is a complex parameter.  With respect to the (additive) basis $1,x, x^2,\dots,x^n$ of $H^\ast(\C P^n;\C)$, the matrix of quantum multiplication by $x$ is
\[
\om(q)=
\bp
 & & & q\\
 1 & & & \\
  & \ddots & & \\
   & & 1 &
\ep
\]
The connection form $\tfrac1\hbar \om(q)\tfrac{dq}{q}$ plays a fundamental role in quantum cohomology theory; in our current notation $z=q$ and $\la=\hbar$, so we take
$\tfrac1\la \eta(z)dz=
\tfrac1\hbar \om(q)\tfrac{dq}{q}$.  Thus, the holomorphic data for the quantum cohomology of $\C P^n$ is
given by $p_0=1,p_1=z^{-1},\dots,p_n=z^{-1}$. The first two rows of Table 3 are the cases $n=3, n=4$. We use the notation of \cite{GuSaXX} for the orbifold quantum cohomology of weighted projective spaces.

The Milnor ring $\C[x,q]/(x^n-q)$
of the unfolding $\tfrac1{n+1} x^{n+1} - tx$ of the  $A_n$ singularity is used in the same way: the matrix of multiplication by $x$ is taken as the matrix $\eta$. The connection is taken as $\tfrac1\la \eta(z)dz$. This gives the holomorphic data shown in Table 4.

Evidently these matrices are not canonical as they depend on  choices of bases.  The exponents $k_0,\dots,k_n$ are to some extent canonical (they are determined up to a change of variable\footnote{We remark also that, for the weighted projective spaces, a change of variable $z\mapsto z^{\frac1N}$ renders all the exponents $k_0,\dots,k_n$ integral, without violating the conditions
$k_i\ge -1$ and $\sum_{i=0}^n k_i > -(n+1)$. This does not affect the values of $\ga_0,\ga_1$.}
$z\mapsto z^k$ by the grading of the cohomology ring), but the coefficients $c_0,\dots,c_n$ may be varied by scaling the basis elements independently. 

However, 
what is significant is that 

(1) there exists holomorphic data for each of the examples in Tables 3 and 4 with the properties
$k_i\ge -1$ for all $i$, and $\sum_{i=0}^n k_i > -(n+1)$, and

(2) for all except one example --- the case of $\P(2,3)$, where $\ga_0$ and $\ga_1$ have opposite signs --- the corresponding values of $\ga_0$ and $\ga_1$ satisfy the hypotheses of Theorem \ref{mainresult}.

\begin{table}[ht]
\renewcommand{\arraystretch}{1.5}
\begin{tabular}{cc||c|c|c|c|c|c|c||c|c}
$l_1$ & $l_2$ & space & $p_0$ & $p_1$ & $p_2$
& $p_3$ & $p_4$ & $p_5$ & $\gamma_0$ & $\gamma_1$
 \\
\hline
$4$ &  & $\C P^3$ & $1$ & $z^{-1}$ & $z^{-1}$ & $z^{-1}$ & & & 
${3}$
& 
${1}$
\\
\hline
$5$ &  & $\C P^4$ & $1$ & $z^{-1}$ & $z^{-1}$ & $z^{-1}$ &$z^{-1}$ & & 
${4}$
& 
${2}$
\\
\hline
$1$ & $4$ &  &  &  &  &  & & & 
&
\\
\hline
$1$ & $5$ &  &  &  &  &  & & & 
&
\\
\hline
$2$ & $2$ &$\P(1,3)$ & $\tfrac13 z^{-\tfrac23}$ & 
$z^{-1}$ & $\tfrac13 z^{-\tfrac23}$ & $\tfrac13 z^{-\tfrac23}$ & & & $\frac13$
& 
${1}$
\\
\hline
$2$ & $3$  &$\P(1,4)$ & $\tfrac14 z^{-\tfrac34}$ & 
$z^{-1}$ & $\tfrac14 z^{-\tfrac34}$ & $\tfrac14 z^{-\tfrac34}$ & $\tfrac14 z^{-\tfrac34}$ & & $\frac12$
& 
${1}$
\\
&   &$\P(2,3)$ & $\tfrac13 z^{-\tfrac23}$ & 
$z^{-1}$ & $\frac16 z^{-\tfrac23}$ & $\frac13 z^{-\tfrac56}$ & $\frac12 z^{-\tfrac56}$ & & $-\frac13$
& 
${1}$
\\
\hline
$3$ & $2$&$\P(1,1,3)$ & $\tfrac13 z^{-\tfrac23}$ & 
$z^{-1}$ & $z^{-1}$ & $\tfrac13 z^{-\tfrac23}$ & $\tfrac13 z^{-\tfrac23}$& & $2$
& 
${2}$
\\
\hline
$3$ & $3$ &$\P(1,1,4)$ & $\tfrac14 z^{-\tfrac34}$ & 
$z^{-1}$ & $z^{-1}$ & $\tfrac14 z^{-\tfrac34}$ & $\tfrac14 z^{-\tfrac34}$ & $\tfrac14 z^{-\tfrac34}$& $1$
& 
${2}$
\\
 &   &$\P(1,2,3)$ & $\frac13 z^{-\tfrac23}$ & 
$z^{-1}$ & $z^{-1}$ & $\frac16 z^{-\tfrac23}$ & $\frac13 z^{-\tfrac56}$ & $\frac12 z^{-\tfrac56}$ &  $0$
& 
${2}$
\\
\hline
$4$ & $1$ &$\P(1,1,1,2)$ & $\tfrac12 z^{-\tfrac12}$ & 
$z^{-1}$ & $z^{-1}$ & $z^{-1}$ & $\tfrac12 z^{-\tfrac12}$ & & $1$
& 
${1}$
\\
\hline
$5$ & $1$ &$\P(1,1,1,1,2)$ & $\tfrac12 z^{-\tfrac12}$ & 
$z^{-1}$ & $z^{-1}$ & $z^{-1}$ & $z^{-1}$ & $\tfrac12 z^{-\tfrac12}$ & $2$
& 
${2}$
\end{tabular}
\bigskip
\caption{}
\end{table}

\begin{table}[ht]
\renewcommand{\arraystretch}{1.5}
\begin{tabular}{cc||c|c|c|c|c|c|c||c|c}
$l_1$ & $l_2$ & singularity & $p_0$ & $p_1$ & $p_2$
& $p_3$ & $p_4$ & $p_5$ & $\gamma_0$ & $\gamma_1$
 \\
\hline
$4$ &  & $A_4$ & $z$ & $1$ & $1$ & $1$ & & & 
$\tfrac35$
& 
$\tfrac15$
\\
\hline
$5$ &  & $A_5$ & $z$ & $1$ & $1$ & $1$ &$1$ & & 
$\tfrac23$
& 
$\tfrac13$
\end{tabular}
\bigskip
\caption{}
\end{table}

\begin{corollary}\label{conclusion}
To each of the examples in Tables 3 and 4, except\footnote{It seems likely that
there is a globally smooth solution corresponding to $\P(2,3)$ as well. However, the method of section \ref{solutions} does not apply in this case.}
 for $\P(2,3)$, there is associated a unique smooth solution of the \stoda on $\C\setminus\{0\}$, i.e.\  a \ll globally smooth\rr  $tt^\ast$ structure.
\end{corollary}

A natural problem is to clarify the meaning of \ll associated\rr in the above statement.  Unfortunately this is not a straightforward matter. 

Certainly we can give an explicit algorithm which relates the holomorphic data to the solution of the \sstoda.  When $k_i\ge 0$, this can be read off from the calculations in this section in the following way.  
First, the proof of Theorem \ref{correspondence} 
(formulae (\ref{2w3}),(\ref{2w0}))
shows that any holomorphic data $p_i=c_iz^{k_i}$ (with $k_i\ge 0$) produces a solution of the \stoda which is defined near $t=0$ and satisfies
$w_i(t)=\ga_i\log\vert t\vert + \al_i + O(t)$.  The constants $\ga_i$ are given explicitly in terms of $k_0,\dots,k_n$ and the constants $\al_i$ are given explicitly in terms of $k_0,\dots,k_n$ and $c_0,\dots,c_n$. 

For any given field-theoretic holomorphic data $p_i=\hat c_i z^{k_i}$  for which the $\gamma_i$ satisfy 
$0\le\ga_0<2+\ga_1$,  $0\le \ga_1< 2/b$, we know by Theorem \ref{mainresult} that there exists a solution  of the \stoda which is smooth on 
$\C\setminus\{0\}$.  By the method of part (ii) above, from this solution we obtain holomorphic data
of the form $p_i=c_iz^{k_i}$ with $k_i\ge 0$.  The constants $c_i$ differ
from the constants $\hat c_i$ in general, but we may adjust the \ll holomorphic data to solution\rr correspondence  by using the normalization  $L(0)=\diag(a_0,\dots,a_n)$
instead of $L(0)=I$, for suitable $a_0,\dots,a_n$, to ensure that $c_i=\hat c_i$.  In integrable systems theory, this adjustment is known as a dressing transformation. A similar analysis can be carried out in the case $k_i\ge -1$. 

The ad hoc appearance of this last step has two sources: the holomorphic data of the field-theoretic examples is (as explained earlier) not canonical, and neither is the correspondence between holomorphic data and solutions of the \stoda (it depends on the choice of $L$ and the normalization of the Iwasawa factorization).  Of course the field-theoretic examples themselves are canonical objects, and so are the solutions of the \sstoda, so the problem is to find the right context for a canonical correspondence.   In the case $k_i\ge0$ described above, this would give a direct computation of the constants $\al_i$ from appropriate holomorphic data (by Theorem \ref{mainresult}, the $\al_i$ are determined uniquely by the $\ga_i$ in the case of a solution which is smooth on $(0,\infty)$, although our method does not give a way to compute them).

The theory of \cite{He03} may provide a way to accomplish this in general, using certain holomorphic connections and their monodromy as holomorphic data.
Some results are known already for field-theoretic examples, where intrinsic holomorphic data is available \ll from geometry\rrr.  In the case of quantum cohomology, it was shown in \cite{IrXX} that
mirror symmetry provides a direct route to the solutions of the $tt^\ast$-equations.  In \cite{KaKoPa07} a similar idea was proposed, using language closer to that of \cite{He03} but again based on mirror symmetry. In particular, both \cite{IrXX} and 
\cite{KaKoPa07} produce the \ll correct\rr solution of the \stoda corresponding to $\C P^n$ for arbitrary $n$ (without proving that this solution is smooth on $\C\setminus\{0\}$, however).

\section{Appendix: A rapid derivation 
of the solutions of the two-dimensional 
Toda lattice}\label{appendix}

This section is intended to be a self-contained explanation of the DPW method which gives formulae of Liouville-type or Weierstrass-type for solutions of the (periodic or open) two-dimensional Toda lattice. 

Some comments on the literature are appropriate before we begin. We refer to \cite{Gu97} for elementary information on primitive harmonic maps,  loop groups, and integrable systems.  Only special cases of the DPW method can be found there, however. The DPW method for harmonic maps into symmetric spaces was developed in \cite{DoPeWu98}, and extended to primitive harmonic maps in \cite{BuPe94} and \cite{DMPW97}, but with an emphasis on harmonic maps of finite type (where the DPW potential unfortunately does not appear very naturally in its \ll normalized\rr form $\tfrac1\la\eta(z)dz$).  The relation between the Toda lattice and primitive harmonic maps was explained in \cite{BoPeWo95}, but without using the DPW method.  In view of this, we find it necessary to gather together here some basic facts.

A \ll normalized DPW potential\rr for the Toda lattice is a matrix-valued 
$1$-form $\frac1\la\, \eta(z)\, dz$, where $\eta:U\to \g_{-1}$ is a holomorphic map on some open subset $U\in\C$.
For any $i\in\Z$, $\g_i$ means the $\oomi$-eigenspace of the automorphism of $\slpc$ given by
\[
\tau(X)=d_{n+1}^{-1} X d_{n+1}
\]
where
\[
d_{n+1}=\diag(1,\oom,\dots,\oomn).
\] 
Thus, 
\[
\eta=
\bp
 & & & p_0\\
 p_1 & & & \\
  & \ddots & & \\
   & & p_n &
\ep
\]
for some holomorphic functions $p_i:U\to\C$. 

From this holomorphic data we can construct a solution to the Toda lattice as follows.  Let $L:U^\pr\to \La\SL_{n+1}\C$ be the solution of the  complex o.d.e.\ system $L^{-1}dL = \frac1\la\, \eta\, dz$, with initial condition $L(z_0)=I$ (for some fixed $z_0\in U$ and some simply connected open neighbourhood $U^\prime$ of $z_0$ in $U$).

Let $L=FB$ be the
Iwasawa factorization of $L$ (see chapter 12 of \cite{Gu97}) with $F(z_0)=I$, $B(z_0)=I$. Since the group $\SU_{n+1}$ is compact, this factorization is possible on the entire domain of $L$, namely $U^\prime$.
Since the function $f=\eta$ satisfies
$\tau(f(\la))=f(\oom \la)$, so do the functions $f=L,F,B$. In particular
$B$ is of the form $B=\sum_{i\ge 0}\la^i B_i$ where
$B_i$ takes values in $\g_{i}$, hence  $B_0=\diag(b_0,\dots,b_n)$. The factorization $L=FB$ is unique if we insist that $b_i>0$ for all $i$. We have $b_0\dots b_n=1$ and $b_i(z_0)=1$ for all $i$.

It follows that $\om = F^{-1}dF=F^{-1}F_z dz+F^{-1}F_{\zbar} d\zbar$ must be of the form $\calA dz + \calB d\zbar$ where 
\[
\calA=
\bp
a_0 & & & \\
 &  a_1 & & \\
  & & \ddots & \\
   & &  & a_n
\ep
+
\frac1\la
\bp
 & & & A_0\\
 A_1 & & & \\
  & \ddots & & \\
   & & A_n &
\ep
\]
for some smooth functions $a_i,A_j:U^\prime\to\C$, with
$\calB=-\calA^\ast$.    

From the $\la^{-1}$ terms of $\calA=F^{-1}F_z=(LB^{-1})^{-1}(LB^{-1})_z = \tfrac1\la B\eta B^{-1} + B(B^{-1})_z$, we obtain
\begin{equation}\label{Ai}
A_i=p_i b_i/b_{i-1}
\end{equation}
and similarly from the diagonal terms of $F^{-1}F_{\zbar}$ we obtain
\begin{equation}\label{ai}
a_i=(\log b_i)_z.
\end{equation}
Since $\om = F^{-1}dF$, we have the zero curvature equation $d\om+\om\wedge\om=0$, which gives an additional equation
\begin{equation}\label{aAi}
(a_i)_{\zbar} + (\bar a_i)_z
 = \vert A_{i+1}\vert^2 - \vert A_{i}\vert^2.
\end{equation}

If we write
\[
w_i=\log b_i,
\]
then we obtain the \ll DPW form\rr of the Toda lattice:
\[
 2(w_i)_{\zzb}=\vert p_{i+1}\vert^2e^{2(w_{i+1}-w_{i})} - \vert p_{i}\vert^2e^{2(w_{i}-w_{i-1})}.
\]
This is not yet the standard form of the Toda lattice. 
However, if we redefine $w_i$ by
\begin{equation}\label{wi}
w_i=\log b_i - \log \vert h_i\vert
\end{equation}
where $h_0,\dots,h_n$ are any holomorphic functions, then we obtain
\begin{equation}\label{flexible}
2(w_i)_{\zzb}=
\vert \nu_{i+1} \vert^2 e^{2(w_{i+1}-w_{i})} 
- \vert \nu_{i} \vert^2 e^{2(w_{i}-w_{i-1})}
\end{equation}
where $\nu_i=p_i h_i/h_{i-1}$. We thus gain the freedom to modify the coefficients of by
choosing various $h_0,\dots,h_n$.

For example, let us choose $h_0,\dots,h_n$ such that all $\nu_i$ are equal, say $\nu_i=\nu$ for all $i$  (e.g.\ $h_0=1$ and 
$h_i=\nu^i/(p_1\dots p_i)$ for $i=1,\dots,n$).  Necessarily, $\nu^{n+1}=p_0\dots p_n$. If we introduce a new complex variable $t$ by the formula
$dt/dz=\nu$, then we obtain the standard form of the periodic Toda lattice:
\[
2(w_i)_{t\bar t}=e^{2(w_{i+1}-w_{i})} - e^{2(w_{i}-w_{i-1})}.
\]
If some of the $p_i$ are identically zero, then the equations (\ref{flexible}) uncouple and after a change of variable we obtain open Toda lattices.
If precisely one of $p_0,\dots, p_n$ is identically zero, say
$p_0$, we obtain the standard form of the open Toda lattice.

Thus there is some flexibility in 
the coefficients of the Toda lattice, but not their signs.  In terms of the variables $u_i=2(w_i-w_{i-1})$ we obtain
\[
(u_i)_{z\zbar} = \vert \nu_{i+1}\vert^2 e^{u_{i+1}} - 2\vert \nu_{i}\vert^2e^{u_{i}}+\vert \nu_{i-1}\vert^2e^{u_{i-1}}.
\]
This is
\[
 (u_i)_{\zzb} = - \smallsum_{j=0}^{n} \,k_{ij}\vert \nu_{j}\vert^2 e^{u_j},
\]
i.e.\ we simply post-multiply the Cartan matrix by the positive diagonal matrix $\diag(\vert \nu_{0}\vert^2,\dots,
\vert \nu_{n}\vert^2)$.

Let us summarise the above construction:

\begin{theorem}\label{holtotoda}
From any  holomorphic $p_0,\dots,p_n,h_0,\dots,h_n:U\to\C$, any simply connected open subset $U^\pr$ of $U$, and any point $z_0\in U^\prime$, the construction above produces  functions $w_0,\dots,w_n:U^\pr\to\R$ which satisfy 
\begin{equation}\label{flexible2DTL}
 2(w_i)_{\zzb}=\vert \nu_{i+1}\vert^2e^{2(w_{i+1}-w_{i})} - \vert \nu_{i}\vert^2e^{2(w_{i}-w_{i-1})}
\end{equation}
where $\nu_i=p_i h_i/h_{i-1}$. 
In these equations, $i$ is interpreted modulo $n+1$.
\end{theorem}

In the other direction, from any solution 
$\t w_0,\dots,\t w_n$ of 
\[
2(\t w_i)_{\zzb}=
\vert \nu_{i+1} \vert^2 e^{2(\t w_{i+1}-\t w_{i})} 
- \vert \nu_{i} \vert^2 e^{2(\t w_{i}-\t w_{i-1})}
\]
we can retrace the above steps to obtain holomorphic data $\t p_0,\dots,\t p_n$, providing we fix, once and for all, the 
holomorphic functions $p_0,\dots,p_n$ and $h_0,\dots,h_n$ as \ll reference data\rrr.
 Namely, using equations  (\ref{wi}), (\ref{Ai}), (\ref{ai})
successively,  we introduce
\[
\t b_i=\vert h_i\vert e^{\t w_i}
\]
and obtain $\t A_i=p_i \t b_i/\t b_{i-1}$,
$\t a_i=(\log \t b_i)_z$.
By equation (\ref{aAi}),  $\t \om =\t\calA dz + \t\calB d\zbar$ must satisfy $d\t\om + \t\om\wedge\t\om=0$. Hence there is a unique $\t F$ such that $\t F^{-1}d\t F=\t\om$ and $\t F(z_0)=I$. In a neighbourhood of $z_0$, $\t F$ admits a Birkhoff factorization (see chapter 12 of \cite{Gu97}), which we can write in the form
$\t F=\t L\t B^{-1}$, with $\t L(z_0)=I$, $\t B(z_0)=I$.  This gives
\[
\t\eta=
\bp
 & & & \t p_0\\
 \t p_1 & & & \\
  & \ddots & & \\
   & & \t p_n &
\ep
\]
with $\t L^{-1}d\t L=\tfrac1\la \t \eta dz$.  Thus, we have the following \ll converse\rr to Theorem \ref{holtotoda}:

\begin{theorem}\label{todatohol}  Fix holomorphic
$p_0,\dots,p_n, h_0,\dots,h_n:U\to\C$ and define $\nu_i=p_ih_i/h_{i-1}$.  Let $U^\pr$ be any
simply connected open subset of $U$, and let $z_0\in U^\prime$.  From any solution $\t w_0,\dots,\t w_n$ of 
\[
2(\t w_i)_{\zzb}=
\vert \nu_{i+1} \vert^2 e^{2(\t w_{i+1}-\t w_{i})} 
- \vert \nu_{i} \vert^2 e^{2(\t w_{i}-\t w_{i-1})},
\]
on $U^\prime$, 
the construction above produces holomorphic 
functions $\t p_0,\dots,\t p_n$ on a neighbourhood of $z_0$ in $U^\prime$. 
\end{theorem}

\no In general, we cannot conclude that $\t p_0,\dots,\t p_n$ are defined on $U^\prime$, as the Birkhoff factorization may not exist at every point of $U^\prime$. 

The construction of Theorem \ref{todatohol} is not the inverse of the construction of Theorem \ref{holtotoda}. However, in conjunction with the change of variable $t=\int\!\nu\,dz$ (in the case $\nu=\nu_0=\cdots=\nu_n$) it gives a method of producing holomorphic data from a solution of the Toda lattice. In section \ref{tt} we consider a very restricted situation where the constructions are essentially inverse to each other.

{\em

\noindent
Department of Mathematics and Information Sciences\newline
   Faculty of Science and Engineering\newline
   Tokyo Metropolitan University\newline
   Minami-Ohsawa 1-1, Hachioji, Tokyo 192-0397\newline
   JAPAN
   
   \noindent
Taida Institute for Mathematical Sciences\newline
Department of Mathematics  \newline
National Taiwan University \newline
Taipei 10617\newline
TAIWAN

}


\begin{thebibliography}{99}

\bibitem{Bar01}
S.~Barannikov,
\emph{Quantum periods. I. Semi-infinite variations of Hodge structures},
Internat. Math. Res. Notices
\textbf{2001-23}
(2001),
1243--1264.

\bibitem{BoIt95}
A.~Bobenko and A.~Its,
\emph{The Painlev\'e III equation and the Iwasawa decomposition},
Manuscripta Math.
\textbf{87}
(1995),
369--377.

\bibitem{BoPeWo95}
J.~Bolton, F.~Pedit, and L.~Woodward,
\emph{Minimal surfaces and the affine Toda field model},
J. reine angew. Math.
\textbf{459}
(1995),
119--150.

\bibitem{BuPe94}
F.~E.~Burstall and F.~Pedit,
\emph{Harmonic maps via Adler-Kostant-Symes theory},
Harmonic Maps and Integrable Systems,
eds. A.~P.~Fordy and J.~C.~Wood,
Aspects of Math. E23,
Vieweg,
1994,
pp.~221--272. 

\bibitem{CeVa91}
S.~Cecotti and C.~Vafa,
\emph{Topological---anti-topological fusion},
Nuclear Phys. B 
\textbf{367}
(1991),
359--461.

\bibitem{CeVa92}
S.~Cecotti and C.~Vafa,
\emph{Exact results for supersymmetric $\si$ models},
Phys. Rev. Lett.
\textbf{68}
(1992),
903--906
(hep-th/9111016).

\bibitem{CeVa92a}
S.~Cecotti and C.~Vafa,
\emph{On classification of $N=2$ supersymmetric theories},
Commun. Math. Phys.
\textbf{158}
(1993),
569--644
(hep-th/9211097 ).

\bibitem{DoGuRoXX}
J.~F.~Dorfmeister, M.~A.~Guest, and W.~Rossman,
\emph{The $tt^\ast$ structure of the quantum cohomology of $\C P^1$ from the viewpoint of differential geometry},
Asian J. Math., to appear
(arXiv:0905.3876).

\bibitem{DMPW97}
J.~Dorfmeister, F.~Pedit, I.~McIntosh, and H.~Wu,
\emph{On the meromorphic potential for a harmonic surface in a $k$-symmetric space},
manuscripta math.
\textbf{92}
(1997),
143--152.

\bibitem{DoPeWu98}
J.~Dorfmeister, F.~Pedit, and H.~Wu,
\emph{Weierstrass type representations of harmonic maps into symmetric spaces},
Comm. Anal. Geom.
\textbf{6}
(1998),
633--668.

\bibitem{Du93}
B.~Dubrovin,
\emph{Geometry and integrability of 
topological-antitopological fusion},
Comm. Math. Phys.
\textbf{152}
(1993),
539--564.

\bibitem{FIKN06}
A.~S.~Fokas, A.~R.~Its, A.~A.~Kapaev, and V.~Y.~Novokshenov,
\emph{Painlev\'e Transcendents: The Riemann-Hilbert Approach},
Mathematical Surveys and Monographs 128,
Amer. Math. Soc.,
2006.

\bibitem{Gu97}
M.~A.~Guest,
\emph{Harmonic Maps, Loop Groups, and Integrable Systems},
LMS Student Texts 38,
Cambridge Univ. Press,
1997.

\bibitem{Gu08}
M.~A.~Guest,
\emph{From Quantum Cohomology to Integrable Systems},
Oxford Univ. Press,
2008.

\bibitem{GuSaXX}
M.~A.~Guest and H.~Sakai,
\emph{Orbifold quantum D-modules associated to weighted projective spaces},
preprint
(arXiv:0810.4236).

\bibitem{He03}
C.~Hertling,
\emph{$tt^\ast$ geometry, Frobenius manifolds, their connections,  and the construction for singularities},
J. reine angew. Math.
\textbf{555}
(2003),
77--161.
     
\bibitem{IrXX}
H.~Iritani,
\emph{Real and integral structures in quantum cohomology I: toric orbifolds},
preprint
(arXiv:0712.2204; see also the updated version
\emph{$tt^\ast$-geometry in quantum cohomology}, 
arXiv:0906.1307).

\bibitem{JaTa80}
A.~Jaffe and C.~H.~Taubes,
\emph{Vortices and Monopoles},
Birkh\"auser,
1980.

\bibitem{JoLiWa05}
J.~Jost, C.-S.~Lin, and G.~Wang,
\emph{Analytic aspects of the Toda system II: bubbling behaviour and existence of solutions},
Comm. Pure Appl. Math.
\textbf{59}
(2006),
526--558.

\bibitem{JoWa01}
J.~Jost and G.~Wang,
\emph{Analytic aspects of the Toda system I: a Moser-Trudinger inequality},
Comm. Pure Appl. Math.
\textbf{54}
(2001),
1289--1319.

\bibitem{KaKoPa07}
L.~Katzarkov, M.~Kontsevich, and T.~Pantev,
\emph{Hodge theoretic aspects of mirror symmetry},
From Hodge Theory to Integrability and TQFT: tt*-geometry,
eds. R.~Y.~Donagi and K.~Wendland,
Proc. of Symp. Pure Math. 78, 
Amer. Math. Soc. 2007,
pp.~87--174 
(arXiv:0806.0107).

\bibitem{LiWa10}
C.-S.~Lin and C.-L.~Wang,
\emph{Elliptic functions, Green functions and the mean field equations on tori},
Ann. Math.
\textbf{172}
(2010),
911-954.

\bibitem{MTW77}
B.~M.~McCoy, C.~A.~Tracy, and T.~T.~Wu,
\emph{Painlev\'e functions of the third kind},
J. Math. Phys.
\textbf{18}
(1977),
1058--1092.

\bibitem{RaSa97}
A.~V.~Razumov and M.~V.~Saveliev,
\emph{Lie Algebras, Geometry, and Toda-type Systems},
Cambridge Lecture Notes in Physics 8,
Cambridge Univ. Press,
1997.

\bibitem{RaSa97a}
A.~V.~Razumov and M.~V.~Saveliev,
\emph{Multidimensional Toda type systems},
Theor. Math. Phys.
\textbf{112}
(1997),
999-1022
(hep-th/9609031).

\bibitem{YiYa01}
Y.~Yang,
\emph{Solitons in Field Theory and Nonlinear Analysis},
Springer,
2001.



\end{thebibliography}
\end{document}